\documentclass[12pt]{article}
\usepackage{graphicx}
\usepackage{amsmath,amssymb,setspace,amsfonts,enumitem,fullpage,amsthm,comment,color,bm,nth}
\usepackage[longnamesfirst]{natbib}
\usepackage {palatino}
\usepackage{epsf}
\usepackage[margin=10pt,labelfont=bf,labelsep=endash]{caption}
\usepackage[colorlinks,pagebackref=false]{hyperref}
\usepackage[usenames,dvipsnames]{xcolor}
\usepackage{titlesec}
\usepackage[bottom]{footmisc} %Pushes footnotes to bottom of page
\usepackage{caption}
\usepackage{subcaption}
\usepackage{tikz}
\allowdisplaybreaks
\usetikzlibrary{arrows,backgrounds,calc,decorations.pathreplacing,decorations.markings}
\colorlet{myblue}{blue!80!green}
\colorlet{mybluelight}{myblue!50}
\tikzset{
  > = latex',
  axis/.style    = {very thick},
  aborder/.style = {draw},
  acomp/.style   = {fill=black, fill opacity=0.1},
  rect/.style    = {very thick},
  form/.style    = {font=\scriptsize},
  sm/.style      = {font=\small},
  vsm/.style     = {font=\scriptsize}
}

\definecolor{dark-red}{rgb}{0.4,0.15,0.15}
\definecolor{dark-blue}{rgb}{0.15,0.15,0.75}
\definecolor{medium-blue}{rgb}{0,0,0.5}
\hypersetup{
    pdftitle={Delegation in Veto Bargaining},
    pdfauthor={Kartik, Kleiner, van Weelden},
    citecolor=dark-blue,
    bookmarksnumbered=true,
    urlcolor=BrickRed,
    linkcolor=dark-blue
}

\makeatletter
  \renewcommand\@seccntformat[1]{\csname the#1\endcsname.{\hskip.7em\relax}} %Gets period after section title
\makeatother
\makeatletter
\renewenvironment{proof}[1][\proofname] {\par\pushQED{\qed}\normalfont\topsep6\p@\@plus6\p@\relax\trivlist\item[\hskip\labelsep\bfseries#1\@addpunct{.}]\ignorespaces}{\popQED\endtrivlist\@endpefalse}
\makeatother
\newcommand{\R}{\mathbb{R}}
\renewcommand{\L}{\mathcal{L}}
\DeclareMathOperator*{\Prob}{Prob}

\newtheorem{innercustomcondition}{Condition}
\newenvironment{customcondition}[1]{\renewcommand\theinnercustomcondition{#1}\innercustomcondition}{\endinnercustomcondition}

\newtheorem{corollary}{Corollary}
\newtheorem{lemma}{Lemma}

\newtheorem{proposition}{Proposition}

\theoremstyle{remark}
\newtheorem{remark}{Remark}
\newtheorem{example}{Example}

\theoremstyle{definition}

\titlespacing\section{0pt}{10pt plus 2pt minus 2pt}{4pt plus 2pt minus 2pt} %Tightens up spacing after section title
\titlespacing\subsection{0pt}{6pt plus 2pt minus 2pt}{2pt plus 2pt minus 2pt} %Tightens up spacing after subsection titile
\titlespacing\subsubsection{0pt}{6pt plus 2pt minus 2pt}{0pt plus 2pt minus 2pt} %Tightens up spacing after subsubsection title

\DeclareMathOperator*{\argmax}{arg\,max}

\DeclareMathOperator{\sign}{sign}
\DeclareMathOperator*{\Obj}{Obj}

\renewcommand{\epsilon}{\varepsilon}
\def \union{\cup}
\def \amid{\frac{\overline a + \underline a}{2}}

\def \Reals{\mathbb R}
\def \amid{a_{\mathrm{mid}}}
\def \mode{\mathrm{Mo}}

\let\oldfootnote\footnote
\renewcommand\footnote[1]{\oldfootnote{\hspace{.5mm}#1}}
\newcommand{\E}{\mathbb{E}}

\newcommand{\mailto}[1]{\href{mailto:#1}{\texttt{#1}}} %creates an email command
\newcommand{\citepos}[1]{\citeauthor{#1}'s \citeyearpar{#1}} %creates a possesive cite command
\newcommand{\citeauthorpos}[1]{\citeauthor{#1}'s} %creates a possesive cite command

\renewcommand{\bar}{\overline}

\newcommand{\appendixref}[1]{\hyperref[#1]{Appendix \ref{#1}}}

\def \marker{\hfill $\diamond$}

\begin{document}

\begin{titlepage}

\title{\textbf{Delegation in Veto Bargaining}\thanks{We thank Nageeb Ali, Ricardo Alonso, Andy Daughety, Wouter Dessein, Wiola Dziuda, Alex Frankel, Sanjeev Goyal, Marina Halac, Elliot Lipnowski, Mallesh Pai, Mike Ting, Andy Zapechelnyuk, the Editor (Jeff Ely), four anonymous referees, and various seminar and conference audiences for helpful comments. We gratefully acknowledge financial support from NSF Grants SES-2018948, SES-2018599, and SES-2018983. Bruno Furtado provided excellent research assistance.}\author{Navin Kartik\footnote{Department of Economics, Columbia University.  Email: \mailto{nkartik@columbia.edu}.}  \and Andreas Kleiner\footnote{Department of Economics, Arizona State University.  Email:  \mailto{andreas.kleiner@asu.edu}}  \and Richard Van Weelden\footnote{Department of Economics, University of Pittsburgh.  Email: \mailto{rmv22@pitt.edu}.}}}

\date{April 29, 2021}

\maketitle

\begin{abstract}
%148 words
A proposer requires the approval of a veto player to change a status quo. Preferences are single peaked. Proposer is uncertain about Vetoer's ideal point. We study Proposer's optimal mechanism without transfers. Vetoer is given a menu, or a delegation set, to choose from. The optimal delegation set balances the extent of Proposer's compromise with the risk of a veto. Under reasonable conditions, ``full delegation'' is optimal: Vetoer can choose any action between the status quo and Proposer's ideal action. This outcome largely nullifies Proposer's bargaining power; Vetoer frequently obtains her ideal point, and there is Pareto efficiency despite asymmetric information. More generally, we identify when ``interval delegation'' is optimal. Optimal interval delegation can be a Pareto improvement over cheap talk. We derive comparative statics. Vetoer receives less discretion when preferences are more likely to be aligned, by contrast to expertise-based delegation. Methodologically, our analysis handles stochastic mechanisms.

\end{abstract}

\thispagestyle{empty}
\end{titlepage}

\section{Introduction}
\label{s:intro}

\paragraph{Motivation.} There are numerous situations in which one agent or group can make proposals but another must approve them. Legislatures (e.g., the U.S. Congress) send bills to executives (e.g., the President), who can veto them. Governmental legislation can be struck down as unconstitutional by the judiciary. %The judicial ruling of a lower court can be overturned by a higher court. %A CEO or Board may need shareholders to sign off on certain decisions. 
Prosecutors choose which charges to bring against defendants, but judges and juries decide whether to convict. A real-estate agent can recommend a house to his client, but the client must decide to put in an offer; similarly, a search committee can put forward a candidate, but the organization decides whether to hire her.

\citet{RR:78} present a seminal analysis of such \emph{veto bargaining}. Their framework is one of complete information 
in which a proposer (Congress, government, prosecutor, salesperson, search committee) makes a take-it-or-leave-it proposal to a veto player (President, judiciary, judge/jury, customer, organization). Preferences are single peaked. That is, rather than ``dividing a dollar'', the negotiating parties share some preference alignment. Such preferences are plausible in the contexts mentioned above.%: even a salesperson who wants to sell the most expensive product is not completely opposed with his customer, because the customer may be willing to pay for some quality.

Our paper studies veto bargaining with incomplete information. The proposer is uncertain about the veto player's preferences---specifically, which proposals would actually get vetoed.  %\footnote{It does not matter what the veto player knows about the proposer's preferences.} 
Previous scholars have emphasized this feature's importance; %see \citet{RR:79} 
see \citepos{CM:04} survey. But to our knowledge, our paper is the first that takes a general approach to the issue. We do not assume the proposer is restricted to making a single proposal \citep[e.g.,][]{RR:79}, nor do we fix any particular negotiating protocol \citep[e.g., one round of cheap talk in][]{Matthews:89,FR:20}. Instead, we consider the possible outcomes across all possible protocols by taking a mechanism design approach. Our focus is on identifying the proposer's optimum.

There are at least two reasons this mechanism design approach is of interest. First, it identifies an upper bound on the proposer's welfare. Second, %in the present context, 
as is standard in settings without transfers, any (deterministic) mechanism is readily interpreted and implemented as \emph{delegation}. That is, the proposer simply offers a menu of options; the veto player can select any one or reject them all. %We will see that it is frequently optimal for the proposer to offer a nontrivial menu to the veto player, even though the proposer faces no uncertainty about his own preferences over the options. 
Menus or delegation sets are observed in practice in some applications of our model. Salespeople show customers subsets of products, and search committees put forward multiple candidates for their organization to choose among. In politics, a bill authorizing at most \$$x$ of spending effectively offers an executive who controls the implementing bureaucracy the choice of any spending level in the interval $[0,x]$. Bills can also grant more or less discretion of how to allocate a given level of spending, as encapsulated by former Senator Russ Feingold in the context of the first U.S. coronavirus stimulus package: ``Congress has to decide how much discretion it wants to delegate to executive branch officials.'' (\href{https://www.washingtonpost.com/business/2020/03/22/treasury-coronavirus-senate-corporate-loan/}{Washington} \href{https://www.washingtonpost.com/business/2020/03/22/treasury-coronavirus-senate-corporate-loan/}{Post}, March 22, 2020.)

\paragraph{Main results.}
We formally study a one-dimensional environment. There is a status quo policy or action, $0$, that obtains if there is a veto. There are two agents. Proposer's ideal action is $1$. Vetoer's ideal action, $v$, is her private information. We assume Vetoer preferences are represented by a quadratic loss function, but allow Proposer to have any concave utility function.\footnote{Proposer's risk attitude is important because he faces uncertainty about the final action. Among deterministic mechanisms, Vetoer's risk attitude is irrelevant, so quadratic loss entails no restriction beyond symmetry around the ideal point.}

Our first result (\autoref{prop:full}) identifies conditions under which it is optimal for Proposer to fully compromise and simply let Vetoer choose her preferred action in the interval $[0,1]$. We call this \emph{full delegation}, because Proposer only excludes options that are, from his point of view, dominated by simply offering his ideal action $1$ no matter Vetoer's ideal point. Intuitively, full delegation is optimal when the specter of a veto looms large; in particular, it is sufficient that the density of Vetoer's ideal point is decreasing on the unit interval. %We explain how full delegation's optimality in this case can be understood from second-order stochastic dominance reasoning. 
Optimality of full delegation is quite striking: despite Proposer having considerable bargaining and commitment power, it is Vetoer who frequently gets her first best. % so long as $v\in[0,1]$. 
This is a telling manifestation of private information's consequences. Since full delegation implies ex-post Pareto efficiency, large information rents can obtain here without generating inefficiency, unlike in most other settings.

%Unlike in most other settings that confer information rents, however, full delegation implies ex-post Pareto efficiency. %, regardless of whether $v\in [0,1]$ or not.

The opposite of full delegation is \emph{no compromise}: Proposer only offers his ideal action, $1$. Of course, Vetoer can veto and choose the status quo $0$. \autoref{prop:no_compromise} gives conditions for optimality of no compromise. It is sufficient, for example, that Proposer has a linear loss function---so, in the relevant region of actions, $[0,1]$, he only cares about the mean action---and the density of Vetoer's ideal point is increasing in this region. This case juxtaposes nicely against the aforementioned decreasing-density condition for full delegation.

Both full delegation and no compromise are boundary cases of \emph{interval delegation}: Proposer offers a menu of the form $[c,1]$. %(Such intervals are without loss of optimality, i.e., if an interval $[a,b]$ is optimal, then so is some interval $[c,1]$.) 
Interval delegation is interesting for multiple reasons. Among them are that such delegation sets are simple to interpret and implement, and they turn out to be tractable for comparative statics.  \autoref{prop:interval} provides conditions under which interval delegation is optimal.  These are met, in particular, when Proposer has a linear or quadratic loss function and Vetoer's ideal point distribution is logconcave (\autoref{c:LQinterval}).

We show that, under reasonable conditions, optimal interval delegation yields a Pareto improvement over singleton proposals, even when cheap talk is allowed and full delegation is not optimal (\autoref{p:pareto}). We trace the intuition to Proposer being more willing to compromise when he can offer Vetoer an interval of options rather than only singletons.

We develop two comparative statics, restricting attention to interval delegation --- either justified by optimality or otherwise. First, what happens when Proposer becomes more risk averse? %in the classic Arrow-Prat \citep{Arrow:65,Pratt:64} sense of absolute risk aversion? 
\autoref{prop:compstats}\ref{p:risk} establishes that Vetoer is given more discretion: the optimal threshold in the interval delegation set (i.e., that denoted $c$ above) decreases. Intuitively, Proposer offers a larger set of options to mitigate the risk of a veto. Second, what about when Vetoer becomes more ex-ante aligned with Proposer? Formally, we consider right shifts in Vetoer's ideal point distribution, in the sense of likelihood ratio dominance. \autoref{prop:compstats}\ref{p:LR} establishes that discretion decreases: the optimal interval delegation threshold increases. Intuitively, this is because Proposer is less concerned that a veto will occur. %It bears highlighting that, by contrast to these results, ``perverse'' comparative statics may obtain absent interval delegation.
Although these comparative statics appear natural, it is the structure of interval delegation that allows us to establish them.

\paragraph{Contrast with expertise-based delegation.} The second comparative static mentioned above contrasts with a key theme of the expertise-based delegation literature following \citet{Holmstrom:77,Holmstrom:84}. In that literature, an agent is given discretion over actions because her private information is valuable to the principal; the principal limits the degree of discretion because of preference misalignment. One version of the so-called Ally Principle says that a more aligned agent receives more discretion. 
\citet{Holmstrom:84} establishes its validity under reasonably general conditions, so long as delegation sets take the form of intervals.\footnote{The comparative static may fail absent interval delegation \citep[e.g.,][]{AM:08}.} In our setting, the reason Proposer gives Vetoer discretion is fundamentally different from that in \citeauthor{Holmstrom:84}: it is not to benefit from Vetoer's expertise; rather, Proposer trades off the risk of a veto with the extent of compromise. (In jargon, our delegator has state-independent preferences, by contrast to the state-dependent preferences in most of the literature following \citeauthor{Holmstrom:84}.) Hence we find less discretion emerging when there is, in a suitable sense, more ex-ante preference alignment.

\paragraph{Applications.} In \autoref{s:applications} we present three applications: product menus offered to a consumer; the legal doctrine of lesser-included offenses; and legislative bills put forward to an executive. These applications illustrate the relevance of our general results and allow us to discuss additional implications. For example, we comment on voluntary disclosure of information by consumers to sellers; welfare consequences of the lesser-included offenses doctrine; and alternative interpretations of how discretion can be granted in the legislative application.

\paragraph{Methodology.} We hope some readers will find our analysis interesting on a methodological level. While it is convenient and economically insightful to describe our substantive results in terms of optimal delegation sets, the formal problem we study is one of mechanism design without transfers. Our analytical methodology builds on the infinite-dimensional Lagrangian approach advanced by \citet{AAW:06} and \citet{AB:13}. Unlike these authors and many others, including the important contributions by \citet{MS:91} and \citet{AM:08}, we also cover stochastic mechanisms.\footnote{A qualification is appropriate: both \citet{AAW:06} and \citet{AB:13} allow for money burning, identifying conditions under which optimal mechanisms do not employ that instrument; see \citet{ABF:18} as well. Stochastic mechanisms are equivalent to money burning for certain preference specifications, but in general they are not equivalent. \cite{AE:17} discuss settings in which money burning can be optimal.} That is, we allow for mechanisms in which Vetoer may choose among lotteries over actions. We view such mechanisms as not only theoretically important, but also relevant in applications. %For instance, in the hiring application mentioned earlier, the search committee (Proposer) can offer the organization (Vetoer) an option of hiring ``the best available candidate with at least five years' work experience in that country''. Both parties would view this option as a lottery over some subset of candidates.
%For instance, consider a President (Proposer) nominating a judge for confirmation from the Senate (Vetoer) to a lifetime appointment.  Some judges have extensive written records or well-identified ideologies, whereas others have a less clear-cut judicial philosophies or doctrines. Both parties would view a nominee of the latter kind as a lottery over the set of judicial types. 
	 For instance, consider a President (Proposer) nominating a judge for confirmation from the Senate (Vetoer) to a lifetime appointment. Both the President and the (pivotal voter in the) Senate have preferences over the ideology of the appointed judge. While some potential nominees have extensive written records or well-identified ideologies, others have less of a record of their own legal opinions and thus could be viewed as a lottery over ideologies.\footnote{For example, while John Roberts had argued several cases before the U.S. Supreme Court, he had served only two years as a judge prior to his nomination to the Court by President George W. Bush in 2005; in 2020, Vice President Mike Pence complained that ``John Roberts has been a disappointment to conservatives''.}

Stochastic mechanisms can sometimes be optimal in our framework. Nevertheless, we establish that our sufficient conditions for full delegation, no compromise, and interval delegation (Propositions \ref{prop:full}, \ref{prop:no_compromise}, and \ref{prop:interval}) ensure optimality of these (deterministic) mechanisms even among stochastic mechanisms. 
Furthermore, by permitting stochastic mechanisms, our sufficient conditions are shown to also be necessary for a class of Proposer's utility functions that include linear and quadratic loss. Our approach to handling stochastic mechanisms should be useful in other delegation problems. %\citet{KM:09} have previously studied stochastic mechanisms---establishing conditions under which deterministic mechanisms are and are not optimal---for expertise-based delegation without a veto option.

Recently, \citet{KZ:19} have introduced \emph{balanced delegation problems}, which are delegation problems in which certain extreme actions or outside options must be included. Our setting fits into their general framework, as one can assume the status quo must be part of the delegation set. \citeauthor{KZ:19} derive a general equivalence between such problems and monotone Bayesian persuasion problems. More concretely, they show how some results from the latter literature \citep[e.g.,][]{Kolo:18,DM:19} can be brought to bear on ``linear'' balanced delegation problems.\footnote{This linearity requires that the utilities of Proposer and (all types of) Vetoer, viewed as a function of the action, have the same curvature.}  Our approach of directly studying the delegation problem is complementary and has some advantages. First, it permits insights absent said linearity: this is most evident in our full delegation result. Second, we believe it provides some more transparent economic intuitions. Third, unlike \citet{KZ:19}, we can address stochastic mechanisms and necessity of our sufficient conditions. At a broader level, note that by contrast to us, \citet{KZ:19} highlight applications concerning expertise-based delegation (i.e., with state-dependent delegator preferences). \citet{Zap:19} applies their methodology to a quality certification problem that, he shows, maps into a delegation problem in which the delegator's preferences are state independent.

%Like us, \citet{AB:19} and \citet{Saran:20} directly analyze delegation problems with an outside option. \citet{AB:19} focus on monopoly regulation \citep{BM:82} without transfers; unlike in our paper, their delegator's preferences are state dependent. \citepos{Saran:20} work is concurrent with ours; his general model subsumes ours, but his analysis places some constraints even among deterministic mechanisms. Section 5 of \citet{Saran:20} is the closest comparison with our analysis, though he only considers linear Proposer preferences there: our results on full delegation (\autoref{prop:full}), no compromise (\autoref{prop:no_compromise}), and interval delegation (\autoref{prop:interval}) are, more or less, stronger than his corresponding results on agent-optimal, take-it-or-leave-it, and minimal-acceptable-action mechanisms. On the other hand, he present examples in which other mechanisms outperform these; see his Figures 6 and 7. Neither \citet{AB:19} nor \citet{Saran:20} consider stochastic mechanisms or necessity of their sufficient conditions, nor do they develop our comparative statics or comparisons with cheap talk.
Like us, \citet{Saran:20} and \citet{AB:21} directly analyze delegation problems with an outside option. Iterations of both papers have developed concurrently with ours. Both papers study frameworks that are more general than ours insofar as they allow for state-dependent delegator preferences. Our approach to finding conditions for the optimality of delegation sets differs from theirs. This is most evident in that neither \citet{Saran:20} nor \citet{AB:21} consider stochastic mechanisms or necessity of their sufficient conditions. Our approach also allows us to deduce more permissive sufficient conditions for our delegation sets than their results would when specialized to state-independent delegator preferences.\footnote{With regards to \cite{Saran:20}: his Section 5 considers state-independent preferences, specifically the analog of our linear loss Proposer utility. Our sufficiency condition for full delegation (\autoref{prop:full}) subsumes his on the agent-optimal mechanism, while our condition for interval delegation (\autoref{prop:interval}) subsumes his on take-it-or-leave-it and minimal-acceptable-action mechanisms. (His take-it-or-leave-it mechanism is actually our interval delegation rather than our no compromise because we make different assumptions on the support of Vetoer's ideal-point distribution.) On the other hand, he presents examples in which there are mechanisms that outperform these; see his Figures 6 and 7. 

With regards to \citet{AB:21}: their cap allocation with potential exclusion corresponds to our interval delegation. In our framework, using their approach (see their Section 4) would correspond to finding sufficient conditions that ensure that given an arbitrary interval threshold---including any suboptimal one---it would be constrained-optimal for Proposer to not further restrict Vetoer's action within the given interval. They show in their Section 5 that this approach is fruitful for their monopoly regulation problem with state-dependent regulator preferences. But in a setting with state-independent Proposer preferences like ours, it would imply strong restrictions that rule out many cases we cover. For example, with linear loss Proposer utility, it would imply that our type density must be decreasing on $[0,1]$, so that only full delegation could emerge as optimal.} Furthermore, our comparative statics and comparisons with cheap talk are distinct.

%\citet{AB:21} are motivated by monopoly regulation \citep{BM:82} without transfers. \citepos{Saran:20} work is concurrent with ours; his general model subsumes ours, but his analysis places some constraints even among deterministic mechanisms. Section 5 of \citet{Saran:20} is the closest comparison with our analysis, though he only considers linear Proposer preferences there: our results on full delegation (\autoref{prop:full}), no compromise (\autoref{prop:no_compromise}), and interval delegation (\autoref{prop:interval}) are, more or less, stronger than his corresponding results on agent-optimal, take-it-or-leave-it, and minimal-acceptable-action mechanisms. On the other hand, he present examples in which other mechanisms outperform these; see his Figures 6 and 7. Neither \citet{AB:19} nor \citet{Saran:20} consider stochastic mechanisms or necessity of their sufficient conditions, nor do they develop our comparative statics or comparisons with cheap talk.

\paragraph{Outline.} The rest of the paper proceeds as follows. \autoref{s:model} presents our model. \autoref{s:main} contains our main results on the conditions for optimality of full delegation, no compromise, and, more broadly, interval delegation. \autoref{s:statics} develops comparative statics and makes comparisons with other mechanisms. \autoref{s:applications} contains our applications. \autoref{s:conclusion} concludes. All proofs are in the \hyperref[s:proofs]{appendices}.

\section{Model}
\label{s:model}

\subsection{Veto Bargaining with Incomplete Information}

We consider a classic bargaining problem between two players, a proposer (he) and a veto player (she), who jointly determine a policy outcome or action $a \in \mathbb{R}$. In a manner elaborated below, Proposer makes a proposal that Vetoer can either accept or reject. If Vetoer rejects, a status-quo action is preserved; we normalize the status quo to $0$.

We assume both players have single-peaked utilities. Proposer's utility is $u(a)$ that is concave, maximized uniquely at $a=1$ (essentially a normalization), and twice differentiable at all $a\neq 1$.\footnote{Permitting a point of nondifferentiability allows the linear loss function $u(a)=-|1-a|$. %(We follow the common abuse of terminology by referring to such a $u$ as a loss function.) 
When we write $u'(1)$ subsequently, it refers to the left-derivative when $u$ is not differentiable at 1.} Unless indicated explicitly, we use `concave', `increasing', etc., to mean `weakly concave', `weakly increasing', etc. We will sometimes invoke a restriction to the following subclass of Proposer preferences, which stipulates a convex combination of the widely-used linear and quadratic loss functions.

\begin{customcondition}{LQ}\label{LQ}
For some $\gamma\in [0,1]$,
\begin{equation*}
u(a)=-(1-\gamma)|1-a|-\gamma(1-a)^2. 	
\end{equation*}
\end{customcondition}

Vetoer's utility is represented by $-l(|v-a|)$, where $l(\cdot)$ is strictly increasing. So her utility is symmetric around the unique ideal point $v$.
For tractability, we assume $l(|v-a|)=(v-a)^2$. A subset of our results will rely only on Vetoer's ordinal preferences, for which the choice of quadratic loss entails no loss of generality given that Vetoer's utility is symmetric around her ideal point. Specifically, Vetoer's ordinal preferences are sufficient when we consider only deterministic mechanisms.%\footnote{\label{fn:quadratic}Furthermore, assuming quadratic loss is equivalent to allowing Vetoer any utility function $va +w(a)$ with $w:\Reals\to \Reals$ differentiable and strictly concave \citep[cf.~][]{AB:13}. To see this, note first that by a normalization, we can assume that $\argmax_{a} v a+w(a)=\{v\}$. The first-order condition then implies that $w'(v)=-v$. Integrating yields $w(v)=-v^2/2+w(0)$.}

A key ingredient of our model is that $v$ is Vetoer's private information. We accordingly refer to $v$ as Vetoer's \emph{type}. It is drawn from a cumulative distribution $F$ whose support is an interval $[\underline v,\overline v]$, where we permit $\underline v=-\infty$ and/or $\overline v=\infty$. We assume $F$ admits a continuously differentiable density $f$, and that $f(\cdot)>0$ on $[0,1]$.   All aspects of the environment except the type $v$ are common knowledge. If $v$ were common knowledge, this model would reduce to that of \citet{RR:78}.

Naturally, it is in Proposer's interests to elicit information from Vetoer about $v$. For example, they might engage in cheap talk communication \citep{Matthews:89}, possibly over multiple rounds, or Proposer might make sequential proposals, and so on. To circumvent issues about exactly how the bargaining ensues, we take a mechanism design approach. Following the revelation principle, we consider direct revelation mechanisms, hereafter simply mechanisms. 

A \emph{deterministic mechanism} is described by a real-valued function $\alpha(v)$, which specifies the action when Vetoer's type is $v$, and must satisfy the usual incentive compatibility (IC) and individual rationality (IR) conditions. IC requires that each type $v$ prefers $\alpha(v)$ to $\alpha(v')$ for any $v'\neq v$; IR requires that each type $v$ prefers $\alpha(v)$ to the status quo $0$. Notice that any deterministic mechanism is equivalent to the Proposer offering a (closed) menu or \emph{delegation set} $A \subseteq \Reals$, and Vetoer choosing an action from $A \cup \{0\}$. We will also consider the more general class of \emph{stochastic mechanisms}, which specify probability distributions over actions for each Vetoer type, with analogous IC and IR constraints to those aforementioned. Stochastic mechanisms are theoretically important because the revelation principle does not justify focusing only on deterministic mechanisms. As noted in the introduction, they may also be relevant for applications. A notable contribution of this paper is to establish conditions under which, despite the absence of transfers, stochastic mechanisms cannot improve upon deterministic ones.\footnote{\autoref{rem:stochastic} below explains why stochastic mechanisms can be optimal; \autoref{eg:stochastic} in \appendixref{s:stochastic} elaborates.  \citet[p.~281]{AM:08} provide a related example in their framework without a veto option; see also \citet[Section 4]{KM:09}.}

The mechanism design approach we take can be viewed as identifying an upper bound on Proposer's welfare. That said, as also mentioned in the introduction, we find the implementation via delegation sets quite realistic in various contexts.

\subsection{Proposer's Problem}

We now formally define Proposer's problem. %To draw an analogy with standard mechanism design in what follows, it may be useful to note that the vetoer's utility, $-(v-a)^2$, is equivalently represented by $av-a^2/2$, whose derivative with respect to $v$ is simply $a$. 
Let $M(\R)$ denote the set of Borel probability distributions on $\R$,\footnote{We endow $M(\R)$ with the topology of weak convergence and the corresponding Borel $\sigma$-algebra.} and $M_0(\R)$ be the subset of distributions with finite expectation and finite variance. Denote by $\delta_a$ the degenerate distribution that puts probability 1 on action $a$.  A stochastic mechanism---or simply a mechanism without qualification---is a measurable function $m:[\underline v,\overline v]\rightarrow M_0(\R)$, with $m(v)$ being the probability distribution over actions for type $v$.\footnote{There is no loss in restricting attention to $M_0(\R)$ instead of $M(\R)$ because no type would choose a lottery with infinite mean or variance, given that the status quo is available.} 
To reduce notation, for any deterministic mechanism $\alpha:[\underline v,\overline v]\to \Reals$, we also denote the mechanism $v\mapsto \delta_{\alpha(v)}$ by $\alpha$. For any integrable function $g:A\to \Reals$, let $\E_{m(v)}[g(a)]$ denote the expectation of $g(a)$ when $a$ has distribution $m(v)$.  We only consider mechanisms $m$ for which $v\mapsto \E_{m(v)}[a]$ is integrable. Define the subset of mechanisms
\[\mathcal{S}:= \big\{m: [\underline v,\overline v]\rightarrow M_0(\R)\ | \  m(0)=\delta_0 \text{ and } \forall v<v':\E_{m(v)}[a] \le \E_{m(v')}[a]\big\}.\]
That is, $\mathcal{S}$ consists of mechanisms in which type $0$ gets the status quo and a higher type receives a higher expected action. The first requirement is implied by IR, since Vetoer can always choose the status quo. The second is implied by IC, since Vetoer's utility $-(v-a)^2$ is equivalently represented by $av-a^2/2$; singlecrossing difference in $(a,v)$ yields monotonicity of $\E_{m(v)}[a]$ in $v$ from standard arguments (elaborated in fn.~\ref{fn:envelope} below).

Proposer's problem is:
\begin{align}\label{e:relaxed2}
&\max_{m\in \mathcal{S}} \int \E_{m(v)}[u(a)] \mathrm dF(v)\tag{P} \\
&\text{s.t. }\E_{m(v)}\left[av  - {a^2}/{2}\right] = \int_0^{v} \E_{m(s)}[a] \mathrm ds \hspace{.5cm} \forall v\in[\underline v,\overline v].\tag{IC-env}\label{eq:stochastic_envelope}
\end{align}
%It is without loss to restrict attention to mechanisms in $\mathcal S$ because of IC and IR---in particular, mechanisms in $\mathcal{S}$ satisfy IC's implication that the expected action must be increasing in type---and 
As noted above, it is without loss to restrict attention to mechanisms in $\mathcal S$. The constraint \eqref{eq:stochastic_envelope} captures the additional content of IC, beyond monotonicity, via an analog of the standard envelope formula.\footnote{\label{fn:envelope}Formally, using quadratic utility, IC requires $\forall v,v'$, $\E_{m(v)}[av-a^2/2] \geq \E_{m(v')}[av-a^2/2]$, and IR requires $\forall v$, $\E_{m(v)}[av-a^2/2]\geq 0$.
An IC mechanism $m$ thus satisfies IR if and only if $m(0)=\delta_0$. It follows that $m$ satisfies IC and IR if and only if $m\in \mathcal{S}$ and the envelope formula \eqref{eq:stochastic_envelope} holds. To confirm this, let $\mathbb V(v):= \E_{m(v)}[av-a^2/2]$. Mechanism $m$ is IC if and only if $\mathbb V(v)=\max_{v'}\E_{m(v')}[av-a^2/2]$, which holds if and only if $\mathbb V$ is convex and $\mathbb V(v)=\mathbb V(0)+\int_0^v \E_{m(s)}[a]\mathrm ds$ \citep[Theorem 2]{MS:02}. Consequently, $m$ is IC and IR if and only if $\E_{m(v)}[a]$ is increasing in $v$, \eqref{eq:stochastic_envelope} holds, and $m(0)=\delta_0$.
A technical note: \citepos{MS:02} result applies even when our type space is unbounded because we can effectively restrict attention to types in $[0,1]$ when solving for optimal mechanisms, as elaborated at the outset of \autoref{s:proofs}; moreover, in any IC and IR mechanism, $\E_{m(v)}[a]$ will lie in $[0,2]$ for all $v\in[0,1]$. Therefore, the derivative of Vetoer's utility with respect to her type is bounded in any IC and IR mechanism.}
 Note that since IC requires that no type prefer type $0$'s lottery over its own, and type $0$'s IR constraint requires that it receive action $0$ (captured in $\mathcal S$), every type's IR constraint is implied by type $0$'s. An \emph{optimal mechanism} is a solution to problem \eqref{e:relaxed2}. 

If we restricted attention to deterministic mechanisms, the analogous problem for Proposer would be:
\begin{align*}\label{e:original}
&\max_{\alpha\in \mathcal{A}} \int u(\alpha(v))\ \mathrm dF(v) \tag{D}\\
&\text{s.t. } v \alpha(v) -{\alpha(v)^2}/{2} =\int_0^{v} \alpha(s) \mathrm ds,
\end{align*}
where
$$\mathcal{A}:= \{\alpha:[\underline v,\overline v]\to \Reals \ \vert \ \alpha(0)=0 \text{ and $\alpha$ is increasing}\}.$$ 
 
Any deterministic mechanism $\alpha$ that is IC has a corresponding (closed) delegation set $A_\alpha := \bigcup_{v}\alpha(v)$. Conversely, any delegation set $A$ has a corresponding deterministic mechanism $\alpha_A$ where $\alpha_A(v)$ is the action in $A\cup\{0\}$ that type $v$ prefers the most (with ties broken in favor of Proposer). Note that $\alpha_A$ satisfies IC and IR. While our formal analysis works with mechanisms, it is easier and more economically intuitive to describe our main results, which concern certain deterministic mechanisms, using delegation sets. 
 
 We emphasize some terminology: an \emph{optimal deterministic mechanism} (or an \emph{optimal delegation set}) is a solution to problem \eqref{e:original}. But when we say that a \emph{deterministic mechanism (or delegation set) is optimal}, we mean that it solves problem \eqref{e:relaxed2}, i.e., no stochastic mechanism can strictly improve on it.
 
\subsection{Discussion}

Let us comment on four aspects of our model.

First, veto power is captured via a standard interim IR constraint. An alternative, as in \citet{CJ:09}, would be to allow Vetoer to exercise her veto even after the mechanism determines an action. This is stronger than just ex-post IR because it also strengthens the IC constraint: when type $v$ mimics type $v'$, $v$ may veto a different set of allocations than $v'$ would, and so the action distribution that $v$ evaluates the deviation with is not $m(v')$. Which form of veto power is conceptually appropriate depends on the application. But any IC and IR deterministic mechanism also satisfies the ex-post veto constraint. Hence, the sufficient conditions we provide below for optimality of delegation sets would remain sufficient. 

Second, our model is one of private values: Vetoer's type does not affect Proposer's preferences. This is by way of contrast with the delegation literature initiated by \citet{Holmstrom:84}, in which a principal gives discretion to an agent because of the agent's expertise, i.e., because they have interdependent preferences. We could extend our model and analysis to incorporate this expertise-based delegation or discretion aspect, but one of our main themes is that discretion will emerge even when that is absent, because we instead have veto power.

Third, one might ask why Vetoer relies on Proposer in the first place given private values. Why can't the Vetoer simply choose her ideal point, or require Proposer to offer all the options? At a formal level, we simply take the veto bargaining institution and Proposer's agenda setting power as given, for reasons outside the model.\footnote{\citet{Mylovanov:08} provides a rationale for what he calls ``veto-based delegation''.} But it is plausible in many situations that even though Vetoer knows her preferences, she nevertheless relies on Proposer to provide the options. In the hiring application mentioned earlier, a superior within an organization may well know her preferences, but lacks the time or expertise to find candidates herself. She thus relies on the search committee. In an application elaborated on in \autoref{s:appendices}, a consumer cannot pick among products a salesperson chooses not to make available (and the salesperson can always claim some products are out of stock). A related point in the legislative context is that even when Vetoer knows her spatial/ideological preferences, she cannot simply implement her preferred policy: implementation must be preceded by policy development, which is done by Proposer \citep[cf.~][]{HS:15}.

Fourth, while we have assumed that Proposer has an ideal point of $1$, an equivalent formulation is that Proposer's utility $u(a)$ is globally increasing but he is constrained to only offer actions less than $1$, or there is simply an upper bound on the action space at $1$. %This will become clear momentarily for deterministic mechanisms, but the point also holds for stochastic mechanisms.

\subsection{Preliminary Observations}

Consider delegation sets. Notice first that there is no loss for Proposer in including his ideal action $1$ in the delegation set: for any Vetoer type $v$, either it does not affect the chosen action, or it results in a preferable action. Next, there is no loss for Proposer in excluding actions outside $[0, 1]$: shrinking a delegation set $A$ that contains $1$ to $A \cap [0,1]$ only results in each type choosing an action closer to $1$. As existence of an optimal delegation set follows from standard arguments, we have:

\begin{lemma}
\label{lem:largestset}
There is an optimal delegation set $A$ satisfying $1\in A\subseteq [0,1]$.
\end{lemma}

It would also be without loss to assume that a delegation set contains the status quo, $0$. We don't do so, however, because it is convenient to sometimes describe optimal delegation sets without including $0$.

For any $a\in (0,1)$, the delegation set $A=[a, 1]$ strictly dominates the singleton $\{a\}$ because $A$ results in preferable actions for Proposer when $v>a$. This simple observation highlights the significance of giving Vetoer discretion, despite our model shutting down the expertise-based rationale that the literature initiated by \citet{Holmstrom:84} has focused on.

While Proposer always wants to include action $1$ in the delegation set, he faces a tradeoff when including any action $a\in (0,1)$. Allowing Vetoer to choose such an action $a$ reduces the probability of a veto (or any action less than $a$) but also reduces the probability that Vetoer chooses an action even higher than $a$, which Proposer would prefer to $a$.

\section{Delegation and Optimal Mechanisms}
\label{s:full}
\label{s:main}

\subsection{Full Delegation}
In light of \autoref{lem:largestset}, we refer to the delegation set $[0,1]$ as \emph{full delegation}. Note that full delegation does impose some constraints on Vetoer. But the constraints are minimal: only actions outside the convex hull of the status quo and Proposer's ideal point are excluded. Given the veto-bargaining institution, an outcome of full delegation starkly captures how Vetoer's private information can corrode Proposer's bargaining or agenda-setting power. All Vetoer types in $[0,1]$ obtain their ideal action; no matter Vetoer's type, there is (ex-post) Pareto efficiency, unlike in most other settings that confer information rents. 
Full delegation thus contrasts sharply with the outcome under complete information \citep{RR:78}, in which case Proposer would make Vetoer with ideal point $v<1/2$ indifferent with exercising the veto while getting his own ideal action $1$ from types $v\geq 1/2$. It also contrasts with the outcome under incomplete information were Proposer restricted to making a singleton proposal. In that case the proposal would lead to a veto by some subinterval of types $v\in[0,1]$, hence to ex-post Pareto inefficiency, and all Vetoer types would be weakly worse off, many strictly.\footnote{Action $0$, which can be viewed as a veto, also has positive probability under full delegation when $\underline v<0$.}

It is thus of interest to know when full delegation is optimal. The following quantity concerning the concavity of Proposer's utility will play a key role in our analysis:
\[
\label{kappa}
\kappa:=\inf_{a\in [0,1)} -u^{\prime \prime} (a).
\]
Under \autoref{LQ}, $\kappa = 2\gamma$, which is larger when Proposer's utility puts more weight on its quadratic component relative to its linear component.

\begin{proposition}[Full delegation]\label{prop:full}
Full delegation is optimal if 
\begin{equation}
\label{e:fulldelegation}
\kappa F(v)- u'(v) f(v) \text{ is increasing on } [0,1]. 	
\end{equation}
Conversely, under \autoref{LQ}, full delegation is optimal only if \eqref{e:fulldelegation} holds.
\end{proposition} 

Since $\kappa \geq 0$, $F(\cdot)$ is increasing, and $u'(\cdot)$ is decreasing and nonnegative on $[0,1]$, the proposition directly implies:

\begin{corollary}
\label{c:decr}
Full delegation is optimal if the type density is decreasing on $[0,1]$. 
\end{corollary}

In particular, it is sufficient for full delegation that the type distribution is unimodal with a negative mode. To obtain intuition for the corollary, consider removing any interval $(\underline{a}, \bar{a})$ from a delegation set $A$ that contains $[\underline{a},\overline{a}]$. This change induces Vetoer with type $v \in (\underline{a}, \bar{a})$ to choose between $\underline{a}$ and $\bar{a}$. Due to her symmetric utility function, Vetoer will choose $\underline{a}$ when $v \in \left(\underline{a},\frac{\underline{a}+\overline{a}}{2}\right)$, which harms Proposer, while Vetoer will choose $\overline{a}$  when $v\in \left(\frac{\underline{a}+\overline{a}}{2},\overline{a}\right)$, which benefits Proposer. When the type density is decreasing on $[\underline a,\overline a]$, the former possibility is more likely. In fact, the pruned delegation set induces an action distribution that is second-order stochastically dominated if and only if the type density is decreasing on $[\underline a,\overline a]$.\footnote{\label{fn:SOSD}Let $G_X$ denote the cumulative distribution of the action induced by $A$, $G_Y$ denote that induced by $A\setminus(\underline{a},\overline{a})$, and let $\amid=({\underline{a}+\overline{a}})/{2}$. Since $F$ is the distribution of $v$, it holds that $G_X(a)=G_Y(a)$ for $a\notin (\underline{a},\overline{a}]$, $G_X(a)=F(a)$ on $[\underline{a},\overline{a}]$, and $G_Y(a)=F(\amid)$ for $a\in [\underline{a},\overline{a})$. Consequently, for any $a\in[\underline{a},\amid]$, $G_Y(a)\ge G_X(a)$ and $\int_0^a \left[G_Y(t)-G_X(t)\right]\ \mathrm d t \ge 0$.
Furthermore, for $a\in (\amid,\overline{a}]$,
\begin{align*}
\int_0^a \left[G_Y(t)-G_X(t)\right] \mathrm dt = \int_{\underline{a}}^a \left[F(\amid)-F(t)\right] \mathrm dt \ge \int_{\underline{a}}^{\overline{a}} \left[F(\amid)-F(t)\right] \mathrm dt\ge 0,
\end{align*}
where the last inequality follows from Jensen's inequality because $F$ is concave on $[\underline a,\overline a]$. We conclude that $G_X$ second-order stochastically dominates $G_Y$.

If the type density is not decreasing on $[\underline a,\overline a]$, then second-order stochastic dominance does not hold, but Proposer is hurt by pruning $(\underline a,\overline a)$ if condition \eqref{e:fulldelegation}'s expression $\kappa F(v)-u'(v)f(v)$ is increasing on $[\underline a,\overline a]$.} Since Proposer's utility is concave, he prefers the original delegation set $A$. As any (closed) delegation set contained in $[0,1]$ can be obtained by successively removing open intervals from $[0,1]$, full delegation is an optimal delegation set. While this explanation applies only among delegation sets, 
\autoref{prop:full} implies that \autoref{c:decr} holds even allowing for stochastic mechanisms.

Removing an interval increases the expected action when the type density is increasing, but it also increases the probability of a lower action. Thus, when Proposer is risk averse, it can be optimal to not remove an interval even if the density is increasing on that interval. This explains why condition \eqref{e:fulldelegation} is weaker than $f$ decreasing on $[0,1]$. In general, removing an interval is optimal only if the density is increasing quickly relative to Proposer's risk aversion. This suggests that full delegation is optimal whenever Proposer is sufficiently risk averse.  \autoref{prop:full} allows us to formalize the point using the Arrow-Prat \citep{Arrow:65,Pratt:64} coefficient of absolute risk aversion.

\begin{corollary}
\label{c:risk}
Full delegation is optimal if Proposer is sufficiently risk averse, i.e., if $\inf\limits_{a\in [0,1)} -u''(a)/u'(a)$ is sufficiently large.
\end{corollary}

It should be noted that (regardless of Proposer's risk aversion) optimality of full delegation does require our maintained assumption of $\underline v\leq 0$. Were Vetoer's lowest type $\underline v \in (0,1)$, then full delegation---or even the interval $[\underline v,1]$---would never be an optimal delegation set: it would be strictly worse than $[\min\{2\underline v,1\},1]$.

Readers familiar with \citet{AM:08} may find it helpful to draw a connection between that paper and our \autoref{prop:full}. If we had restricted attention to deterministic mechanisms and assumed that Proposer's utility is a quadratic loss function, then the sufficiency result in \autoref{prop:full} would follow from a result in \cite{AM:08}, even though their model does not have a veto constraint and, as such, highlights expertise-based delegation. To make the connection, we observe that when $u(a)=-(1-a)^2$, condition \eqref{e:fulldelegation} is equivalent to \citeauthorpos{AM:08} ``backward bias'' (p.~264) being convex on $[0,1]$. Their Proposition 2 then implies that if $\{0,1\}$ is contained in the optimal delegation set, then the interval $[0,1]$ is contained in the optimal delegation set. But recall from \autoref{lem:largestset} that in choosing among delegation sets, Proposer need not offer any action outside $[0,1]$ and can offer his ideal point $1$; moreover, he may as well offer the status quo $0$. It follows that full delegation is an optimal delegation set. We emphasize, therefore, that \autoref{prop:full} establishes optimality among more general Proposer preferences and stochastic mechanisms.\footnote{In a model without an outside option, \cite{KM:09} provide sufficient conditions for certain delegation sets to be optimal when stochastic mechanisms are allowed and Proposer has a quadratic loss function.}

Consider now necessity in \autoref{prop:full}. If Proposer has a linear loss utility, then our preceding discussion explains why a delegation set $A$ containing $[\underline a,\overline a]\subseteq [0,1]$ should be pruned to $A\setminus (\underline a,\overline a)$ if the type density is increasing on this interval: the expected action increases. Hence, the converse of \autoref{c:decr} holds for linear loss utility. For quadratic loss utility (and with additional smoothness assumptions), \citepos{AM:08} Proposition 2 implies that $f(v)$ need not be decreasing on $[0,1]$ for full delegation to be optimal, but a weaker condition is necessary: $F(v)-(1-v)f(v)$ must be increasing on $[0,1]$. 
\autoref{prop:full} subsumes these two cases by deducing necessity of condition \eqref{e:fulldelegation} for the linear-quadratic family of utilities (\autoref{LQ}).\footnote{Following our general methodology discussed in \autoref{sec:method}, our proof of necessity uses the availability of stochastic mechanisms. But we can establish that under \autoref{LQ}, \eqref{e:fulldelegation} is necessary even for full delegation to be an optimal delegation set.}

\subsection{No Compromise}
\label{s:nocompromise}

The other extreme from full delegation is \emph{no compromise}: Proposer makes a take-it-or-leave-it offer of his own ideal action, not offering any other action.  Of course, Vetoer can choose the status quo as well. When Proposer has a linear loss utility---or, \emph{a fortiori}, if we had permitted $u(a)$ to be convex on $[0, 1]$---then no compromise is an optimal delegation set whenever the type density $f$ is increasing on $[0,1]$.  This follows from reversing the previous subsection's second-order stochastic dominance argument for optimality of full delegation when $f$ is decreasing. But neither is linear loss utility nor convexity of $F$ on $[0,1]$ required for optimality of no compromise.

\begin{proposition}\label{prop:no_compromise}
Assume \autoref{LQ}. No compromise is optimal if and only if 
\[
(u'(1)+\kappa (1-t))\frac{F(t)- F(1/2)}{t-1/2} \geq (u'(0)-\kappa s)\frac{F(1/2)-F(s)}{1/2- s} \text{ for all $1\geq t>1/2>s\geq 0$}.
\]
\end{proposition}

Under linear loss utility (so $u'(1)=u'(0)=1$ and $\kappa=0$), the condition in \autoref{prop:no_compromise} simplifies to $f(1/2)$ being a subgradient of $F$ at $1/2$ on the domain $[0,1]$. This subgradient condition is weaker than $F$ being convex on $[0,1]$.

\begin{remark}
\label{rem:stochastic}
With linear loss utility, no compromise can be an optimal delegation set (i.e., deterministic mechanism) even if the subgradient condition does not hold. However, there will then be a stochastic mechanism that Proposer strictly prefers. This situation can arise, for example, when the type density is strictly increasing except on a small interval around $1/2$, where it is strictly decreasing. Intuitively, Proposer would like to delegate a small set of actions around $1/2$ to types close to $1/2$, but adding such actions to the no-compromise delegation set is deleterious because it leads to many types above $1/2$ choosing an action close to $1/2$ rather than $1$. By contrast, lotteries with expected value $1/2$ can be used to attract only types close to $1/2$. \autoref{eg:stochastic} in \appendixref{s:stochastic} elaborates. \marker
\end{remark}

We also note that no compromise can be an optimal delegation set even when $u$ does not satisfy \autoref{LQ}. In particular, it can be shown that if no compromise is an optimal delegation set for some $u$, then it is also an optimal delegation set for any utility function that is a convex transformation of $u$.
%TO CHECK WHETHER WE CAN ALSO COVER STOCHASTIC MECHANISMS ABOVE

On the other hand, no compromise is not optimal---not even an optimal delegation set---if Proposer's utility is differentiable at his ideal point $a=1$ (which implies $u'(1)=0$).\footnote{Note that when $u'(1)=0$, the condition in \autoref{prop:no_compromise} fails: its left-hand side is $0$ when $t=1$, while its right-hand side is strictly positive when $s=0$.} The reason is that when $u'(1)=0$, Proposer would strictly benefit from offering a small interval $[1-\epsilon,1]$, or even just the action $1-\epsilon$, instead of only offering action $1$. For, Proposer's decrease in utility from getting an action slightly lower than $1$ is second order, but there is a first-order increase in the probability of avoiding a veto.

\subsection{Interval Delegation}
\label{s:interval}

Both full delegation and no compromise are special cases of \emph{interval delegation}: Proposer offers an interval, and Vetoer chooses an action from either that interval or the status quo. %\footnote{We adopt the convention that for any $x\in \Reals$, $[x,x]$ is identical to $\{x\}$.} 
It follows from \autoref{lem:largestset} that when interval delegation is optimal, there is always an optimal interval of the form $[c, 1]$ for some $c \in [0, 1]$.  One can thus interpret interval delegation as Proposer designating a minimally acceptable option; implicitly, the maximal acceptable option is Proposer's ideal point. Interval delegation, without a status quo, has been a central focus of the prior literature: intervals are simple, tractable, and lend themselves to comparative statics.  Arguably, intervals also map more naturally into proposals likely to emerge in applications.

\begin{proposition}
\label{prop:interval}
The interval delegation set $[c^*,1]$ with $c^*\in [0,1]$ is optimal if
\vspace{-6pt}
\begin{enumerate}[label=(\roman*)]
\item \label{interval1} $\kappa F(v)- u'(v) f(v)$ is increasing on $[c^*,1]$;
\item \label{interval2} $\left(u'(c^*)+\kappa(c^*-t) \right) \frac{F(t)-F(c^*/2)}{t-c^*/2}\geq u'(c^*)\frac{F(c^*)-F(c^*/2)}{c^*/2}$ for all $t\in (c^*/2,c^*]$;
  and
\item \label{interval3} $u'(c^*)\frac{F(c^*)-F(c^*/2)}{c^*/2}\geq \left(u'(0)-\kappa s \right) \frac{F(c^*/2)-F(s)}{c^*/2-s}$ for all $s\in [0,c^*/2)$.
\end{enumerate}
Conversely, under \autoref{LQ}, the delegation set $[c^*,1]$ with $c^*\in (0,1)$  is optimal only if conditions \ref{interval1}, \ref{interval2}, and \ref{interval3} above hold.
\end{proposition}

We discuss sufficiency. The intuition for condition \ref{interval1} in the proposition is analogous to that discussed after \autoref{prop:full}; it ensures that there is no benefit to not fully delegating the interval $[c^*,1]$ taking as given that Vetoer can choose $c^*$. For linear loss utility, the condition reduces to $F$ being concave on $[c^*,1]$. Linear loss utility is also helpful to interpret the other conditions. Conditions \ref{interval2} and \ref{interval3} then simplify to the requirements that the average density from $c^*/2$ to $c^*$ be simultaneously less than that from $c^*/2$ to $t$ for all $t\in (c^*/2,c^*]$ and greater than that from $s$ to $c^*/2$ for all $s \in [0,c^*/2)$. Equivalently, the average density from $c^*/2$ to $c^*$ equals $f(c^*/2)$ and $f(c^*/2)$ is a subgradient of $F$ at $c^*/2$ on the domain $[0,c^*]$. The subgradient condition is analogous to that discussed after \autoref{prop:no_compromise}. (More generally, conditions \ref{interval2} and \ref{interval3} with $c^*=1$ imply the condition of \autoref{prop:no_compromise}.) The additional requirement ensures that the threshold $c^*$ is an optimal threshold. See \autoref{fig:interval}. 

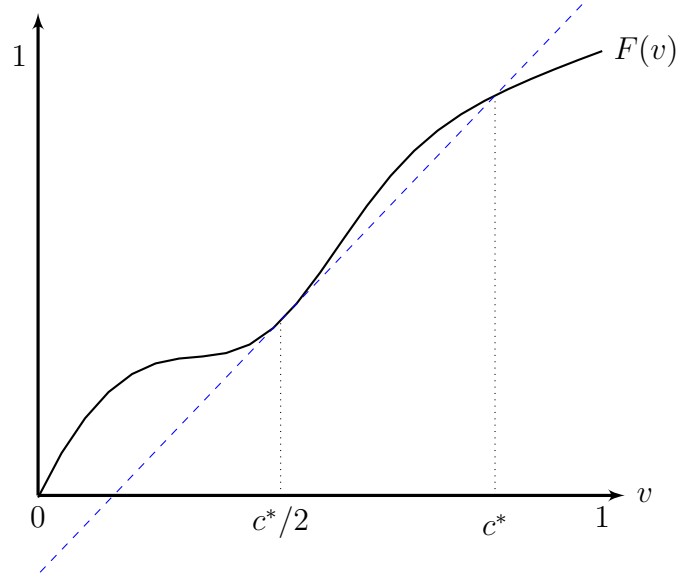
\begin{figure}
\centering
\begin{tikzpicture}[x=1.5cm,y=1.5cm,
    tangent/.style={
        decoration={
            markings,% switch on markings
            mark=
                at position #1
                with
                {
                    \coordinate (tangent point-\pgfkeysvalueof{/pgf/decoration/mark info/sequence number}) at (0pt,0pt);
                    \coordinate (tangent unit vector-\pgfkeysvalueof{/pgf/decoration/mark info/sequence number}) at (1,0pt);
                    \coordinate (tangent orthogonal unit vector-\pgfkeysvalueof{/pgf/decoration/mark info/sequence number}) at (0pt,1);
                }
        },
        postaction=decorate
    },
    use tangent/.style={
        shift=(tangent point-#1),
        x=(tangent unit vector-#1),
        y=(tangent orthogonal unit vector-#1)
    },
    use tangent/.default=1
]
  \draw[axis,->] (0,0) -- (5.2,0) node[right] {$v$};
  \draw[axis,->] (0,0) -- (0,4.28) node[above] {};

  \draw[thick,tangent=0.412] (0,0) plot[domain=0:5] (\x,{ln(pow(\x+1,2.2))-exp(-pow(\x-2,2))}) node[right] {$F(v)$};
  \draw [blue, dashed, use tangent] (-3,0) -- (4,0);
%  \draw [blue, dashed, use tangent] (1,0.2) -- (2.8,.20);
  \draw [black, dotted] (2.15,0) -- (2.15,1.57);
  \draw (2.15,0) node[below] {${c^*}/{2}$};
  \draw [black, dotted] (4.05,0) -- (4.05,3.5);
  \draw (4.05,0) node[below=.1cm] {$c^*$};
  \draw (0,0) node[below] {$0$};
  \draw (5,0) node[below] {$1$};
  \draw (0,3.9) node[left] {$1$};
  %\draw [very thick, brown, step=1.0cm,xshift=0cm, yshift=0cm] (0,0) grid +(5.5,4.5);
\end{tikzpicture}
\caption{Conditions \ref{interval1}---\ref{interval3} of \autoref{prop:interval} for linear loss utility. $F$ is concave on $[c^*,1]$; $f(c^*/2)$ is a subgradient on $[0,c^*]$; and the average density on $[c^*/2,c]$ equals $f(c^*/2)$ because $F(c^*)$ intersects the subgradient.}
\label{fig:interval}
\end{figure}

\begin{remark}\label{rem:unimodal}With linear loss utility, %it can be verified that %if the type distribution is unimodal (i.e., $F$ is first convex and then concave, or equivalently, the density $f$ is single peaked), then interval delegation is optimal. 
interval delegation is optimal when the type distribution is unimodal (i.e., $F$ is first convex and then concave, or equivalently, the density $f$ is single peaked). Either there will be a $c^*\in [0,1]$ satisfying the three conditions of \autoref{prop:interval}, or the condition in \autoref{prop:no_compromise} will be met and no compromise is optimal.\footnote{Let $\mode$ be the (unique and strictly positive, for simplicity) mode of $F$ and let $\Delta(x):=\frac{F(2x)-F(x)}{x}-f(x)$. $F$ being convex-concave implies that letting $c^*/2:= \max\{ x >0 :\Delta(x)= 0\}$, $\Delta(x)\ge 0$ for $x\in(0,c^*/2)$ and $\Delta(x)\le 0$ for $x\in(c^*/2,1]$. Clearly, $c^*/2\le\mode\le c^*$ and hence $f$ is decreasing on $[c^*,1]$. The convex-concave property implies that $f(c^*/2)$ is a subgradient of $F$ at $c^*/2$ on the domain $[0,c^*]$, and if $c^*/2>1/2$ then $f(1/2)$ is a subgradient of $F$ on the domain $[0,1]$.} \marker
\end{remark}

We can extend this observation as follows:

\begin{corollary}
\label{c:LQinterval}
Assume \autoref{LQ}.  Interval delegation is optimal if the type density $f$ is logconcave on $[0,1]$; if, in addition, $f$ is strictly logconcave on $[0,1]$ or Proposer's utility is strictly concave, then there is a unique optimal interval.
\end{corollary}

Recall that logconcavity is stronger than unimodality, but many familiar distributions have logconcave densities, including the uniform, normal, and exponential distributions \citep{BagnoliBergstrom:05}. The proof of \autoref{c:LQinterval} also establishes that under \autoref{LQ} and logconcavity of the type density on $[0,1]$, the set of optimal interval thresholds is connected: if $[c^*_1,1]$ and $[c^*_2,1]$ are both optimal interval delegation sets, then so is $[c^*,1]$ for all $c^*\in[c^*_1,c^*_2]$. Such multiplicity arises under the uniform distribution and linear loss utility. Either strict logconcavity of the type density or strict concavity of Proposer's utility eliminates multiplicity.

Readers familiar with \citet{AB:13} will note from our discussion after \autoref{prop:interval} that condition \ref{interval1} in the proposition plays the same role as condition (c1) on p.~1550 of that paper. In fact, the two conditions are identical, even though our analysis accommodates stochastic mechanisms that cannot be reduced to their money burning.\footnote{Indeed, we conjecture that our methodology, elaborated in \autoref{sec:method}, can be used to show that \citepos{AB:13} conditions ensure optimality of their delegation sets even among stochastic mechanisms.} Conditions \ref{interval2} and \ref{interval3} of \autoref{prop:interval} don't have analogs in \citeauthorpos{AB:13} work, however, because these concern optimality conditions that turn on our status quo.

It bears highlighting that interval delegation is not always optimal. Consider linear loss utility and a single-dipped density (decreasing then increasing) with a dip at $d\in (0,1)$. We claim the  optimal delegation set is now $[0,x] \union \{y,1\}$ for some $x \in [0,d]$ and $y \in [d,1]$. To see why, notice that if any action $a \in [0,d]$ is included, then since the density is decreasing on $[0,a]$, the average action is higher when there are no gaps among actions in $[0,a]$; recall the discussion around \autoref{prop:full}. So $x$ is the maximum action allowed below the dip, i.e., within $[0,d]$ the delegation set takes the form $[0,x]$. On the other hand, if any action $a\in [d,1]$ is included, then since the density is increasing on $[a,1]$, the average action is higher when all actions $(a,1)$ are excluded; recall the discussion around \autoref{prop:no_compromise}. So $y$ is the minimum action allowed above the dip, i.e., within $[d,1]$ the delegation set takes the form $\{y,1\}$. In fact, because the present scenario simply mirrors that discussed in \autoref{rem:unimodal}, it can be shown that it is without loss of optimality to set $y=1$; but this is not needed for the point that interval delegation can be suboptimal.\footnote{For completeness, we note that if the density is strictly single-dipped with the dip at $d>1/2$, then both full delegation and no compromise are strictly suboptimal, which implies that interval delegation is strictly suboptimal.} Furthermore, similar reasoning implies that for certain more complicated type distributions, with multiple peaks and multiple dips, any optimal delegation set with linear loss utility must include some actions in $(0,1)$ while excluding neighborhoods of both $0$ and $1$.

\subsection{Methodology}
\label{sec:method}

Let us outline the idea behind the proofs of Propositions \ref{prop:full}--\ref{prop:interval}. We use a Lagrangian method, as has proved fruitful in prior work on optimal delegation, notably in \citet{AB:13}. However, the presence of a status quo requires some differences in our approach. In particular, while prior work has largely focussed on optimality of connected delegation sets, our \autoref{prop:no_compromise} and \autoref{prop:interval} are effectively concerned with the optimality of disconnected delegation sets because of the status quo. Moreover, our approach provides a simple way to incorporate stochastic mechanisms, which, as already highlighted, are often not addressed in prior work.

Consider the following relaxed version of the optimization problem \eqref{e:original} for deterministic mechanisms:
\begin{align}
&\max_{\alpha\in \mathcal{A}} \int \left(u(\alpha(v)) - \kappa\left[v \alpha(v) -\frac{\alpha(v)^2}{2} - \int_0^{v} \alpha(s) \mathrm ds\right]\right) \mathrm dF(v)\nonumber \tag{R} \label{e:relaxed}\\
&\text{s.t. } v \alpha(v) - \frac{\alpha(v)^2}{2} - \int_0^{v} \alpha(s) \mathrm ds \ge 0 \hspace{.5cm} \forall v\in[\underline v,\overline v]\notag.
\end{align}

Problem \eqref{e:relaxed} is well-behaved because the constraint set is convex and, owing to 
$\kappa\equiv \inf_{a\in [0,1)} -u^{\prime \prime} (a)$, the objective is a concave functional of $\alpha$.
It differs from \eqref{e:original} in two ways. First, the constraint has been relaxed: IC requires the inequality to hold with equality. Second, the objective has been modified to incorporate a penalty for violating IC. Plainly, if $\alpha^*$ is IC and a solution to problem \eqref{e:relaxed}, then it is also a solution to \eqref{e:original}. But we establish (see \autoref{lem:relaxed-stochastic} in \appendixref{s:proofs}) that in this case $\alpha^*$ is also a solution to problem \eqref{e:relaxed2}, i.e., it is optimal among stochastic mechanisms. The idea is as follows. If there were a stochastic mechanism that is strictly better than $\alpha^*$ in problem \eqref{e:relaxed2}, consider the corresponding deterministic mechanism that replaces each lottery by its expected outcome. While this mechanism would not be IC in general, we show that it would both be feasible for problem \eqref{e:relaxed} and would obtain a strictly higher objective value, contradicting the optimality of $\alpha^*$ in \eqref{e:relaxed}. Establishing a higher value relies on the objective in \eqref{e:relaxed} being a concave functional.

The sufficiency results in Propositions \ref{prop:full}--\ref{prop:interval} then obtain from identifying sufficient conditions under which the respective IC mechanisms solve \eqref{e:relaxed}. It is here that our approach involves constructing suitable Lagrange multipliers.

For the necessity results, we first establish in \autoref{l:necessity_linear} in  \appendixref{s:proofs} that under \autoref{LQ}, if a deterministic mechanism $\alpha^*$ solves problem \eqref{e:relaxed2} then it also solves problem \eqref{e:relaxed}. The idea is as follows. Suppose a solution $\alpha$ to problem \eqref{e:relaxed} provides a strictly higher value than $\alpha^*$. We construct a corresponding IC mechanism $m$ such that $\alpha(v)=\E_{m(v)}[a]$ for all $v$. Roughly, monotonicity of $\alpha$ (by definition of the set $\mathcal A$) implies existence of transfers that make $\alpha$ IC in a quasi-linear model; the inequality constraints in \eqref{e:relaxed} mean the transfers can be chosen to be positive (i.e., they can be viewed as money burning); and, because of Vetoer's quadratic utility, positive transfers can be substituted for by the action variance of suitable lotteries. Since $u''(a)=-\kappa$ for $a<1$ under \autoref{LQ}, the condition makes the objective in \eqref{e:relaxed} a linear functional in the relevant domain. We can thus show that mechanism $m$ obtains a strictly higher value than $\alpha^*$ in \eqref{e:relaxed2}, a contradiction.  

We then establish necessity of the conditions in Propositions \ref{prop:full}--\ref{prop:interval} by showing that, unless these conditions are satisfied, the corresponding mechanisms can be strictly improved upon in problem \eqref{e:relaxed}. Here we use the fact that the constraint set in \eqref{e:relaxed} is convex and, therefore, first-order conditions must hold at a solution. More specifically, the (Gateaux) derivative of the objective in the direction of any feasible mechanism must be negative.

\section{Comparative Statics and Comparisons}
\label{s:statics} 

\subsection{Comparative Statics}

We derive two comparative statics, restricting attention to interval delegation. This focus can be justified by implicitly assuming conditions for optimality of interval delegation (\autoref{s:main}), or just because such menus are simple,  tractable, or relevant for applications. 

If Proposer proposes $A=[c, 1]$ with $c\in [0,1]$, then Vetoer chooses $0$ if $v<c/2$, $c$ if $v \in [c/2, c]$, $v$ if $v\in [c, 1]$, and $1$ if $v>1$.  Hence Proposer's expected utility or welfare from $A=[c, 1]$ is
$$W(c):=u(0)F(c/2)+u(c)(F(c)-F(c/2))+\int_{c}^{1}u(v) f(v)\mathrm dv+u(1)(1-F(1)).$$

Differentiating, the first-order condition for $c^*\in (0,1)$ to be an optimal threshold among interval delegation sets is that it must be a zero of
\begin{equation}
\label{e:opt_del}
{2u'(c^*)}\left[{F(c^*)-F(c^*/2)}\right]-{f(c^*/2)}\left[{u(c^*)-u(0)}\right].
\end{equation}

In general there can be multiple optimal thresholds, even among interior thresholds. Accordingly, let the set of optimal thresholds for interval delegation be
\[
C^{*}:=\arg\max_{c\in [0,1]} W(c).
\]

We use the strong set order to state comparative statics. Recall that for $X,Y\subseteq \Reals$, $X$ is larger than $Y$ in the strong set order, denoted $X\geq_{SSO}Y$, 
if for any $x\in X$ and $y\in Y$, $\min\{x,y\}\in Y$ and $\max\{x,y\}\in X$. We say that $C^*$ increases (resp., decreases) if it gets larger (resp., smaller) in the strong set order. Since interval delegation with a lower threshold gives Vetoer a superset of options to choose from, a decrease in $C^*$ corresponds to offering more discretion. It can also be interpreted as Proposer compromising more. As mentioned after \autoref{c:LQinterval}, under \autoref{LQ} and a logconcave type density, $C^*$ is a (closed) interval. In that case a decrease in $C^*$ is equivalent to a decrease in both $\min C^*$ and $\max C^*$. 

Our comparative statics concern Proposer's risk aversion and the ex-ante preference alignment between Proposer and Vetoer. We say that Proposer becomes \emph{strictly more risk averse} if the Arrow-Prat coefficient of absolute risk aversion strictly increases in the relevant region: $-u''(a)/u'(a)$ strictly increases for all $a\in [0,1)$. As is well known, such a change can also be expressed in terms of concave transformations of Proposer's utility. Under \autoref{LQ}, it corresponds to a higher weight on the quadratic term. We say that the two players are \emph{strictly more aligned} if Vetoer's ideal-point density changes from $f$ to $g$ with $g$ strict likelihood ratio dominating $f$ on the interval $[0,1]$: for all $0\leq v_L<v_H\leq 1$, $f(v_H)/f(v_L)<g(v_H)/g(v_L)$.

\begin{proposition}
\label{prop:compstats}
Among interval delegation sets, there is:
\vspace{-8pt}
\begin{enumerate}[label=(\roman*)]
	\item \label{p:risk}more discretion (i.e., $C^*$ decreases) if Proposer becomes strictly more risk averse; and
\item \label{p:LR} less discretion (i.e., $C^*$ increases) if Vetoer becomes strictly more aligned with Proposer.
\end{enumerate}
\end{proposition}

The proof uses the interval delegation structure and monotone comparative statics under uncertainty, specifically \citepos{Karlin:68} variation diminishing property for single-crossing functions and comparative statics from \citet{MS:94}.

The intuition for part \ref{p:risk} of the proposition is simply that greater risk aversion makes Proposer more concerned about a veto, and hence she compromises more. The intuition for part \ref{p:LR} is that greater ex-ante alignment makes Proposer less concerned about a veto, and hence she compromises less.  
Yet, the precise conditions in the proposition are nuanced.  In particular, the stochastic ordering used in our notion of alignment is important: one can construct examples in which, among interval delegation sets, Proposer optimally gives Vetoer strictly more discretion when there is a right-shift in the type density in the sense of either hazard or reversed-hazard rate (both of which are stronger than first-order stochastic dominance but weaker than a likelihood ratio shift).
Furthermore, absent the focus on interval delegation, it is not necessarily clear how to relate changes in delegation sets with the degree of discretion or compromise.

%Should we say more here about importance of intervals?  Do we want to confirm the result for singleton proposals?

It is instructive to contrast part \ref{p:LR} of \autoref{prop:compstats} with the expertise-based delegation literature. The broad finding there is that among interval delegation, greater preference similarity in a suitable sense leads to more discretion \citep[][Theorem 3]{Holmstrom:84}. The difference owes to and highlights the distinct rationales for discretion. In those models, the delegator would like to give the agent discretion to benefit from the agent's expertise; the degree of discretion is limited by the extent of preference misalignment. In our setting, on the other hand, the agent is given discretion only because of her veto power; greater ex-ante preference alignment mitigates that concern.

\begin{example}
\label{eg:normal}
Under \autoref{LQ}, the first-order condition for an optimal interval threshold (i.e., expression \eqref{e:opt_del} equals zero) becomes
$$2(1+\gamma-2 \gamma c^*)\left[F(c^*)-F(c^*/2)\right]=c^*(1+\gamma-\gamma c^*) f(c^*/2).$$

Recall from \autoref{c:LQinterval} that, when combined with boundary conditions, there will be a unique solution for any strictly logconcave type density; moreover, the corresponding interval is then an (unrestricted) optimal mechanism. Given uniqueness, the implicit function theorem can be used to affirm the general comparative statics of \autoref{prop:compstats}; moreover, the first-order condition can also be used to compute numerically the optimal interval threshold for standard distributions. \autoref{fig:compstats} illustrates for Normal distributions. The left panel verifies comparative statics already discussed; note that a higher mean $\mu$ is a likelihood ratio right-shift and hence more alignment.\footnote{When $\mu\geq 1$, a higher $\mu$ can be viewed as shifting Vetoer overall further away to the right of Proposer, but what is relevant is the change of the distribution on the interval $[0,1]$.} The right panel shows comparative statics in the variance of the distribution. We see that there is less discretion when the variance is lower, with the optimal threshold converging, as $\sigma \to 0$, to Proposer's optimal offer, $0.9$, to type $\mu=0.45$. \marker	
\end{example}

\begin{figure}
    \centering
    \begin{subfigure}[t]{0.48\textwidth}
        \centering
        \includegraphics[width=\linewidth]{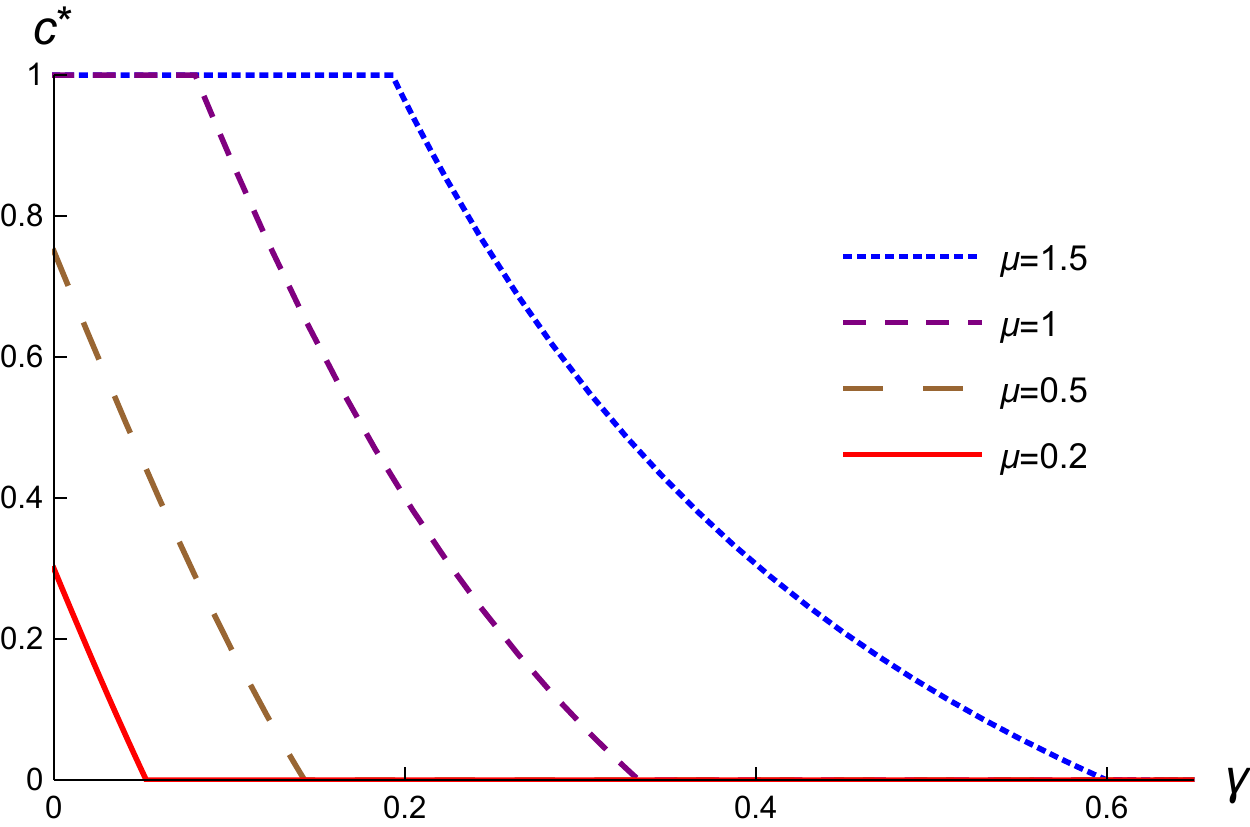} 
        \caption{$\sigma=1$.} \label{fig:statics1}
    \end{subfigure}
    \hfill
    \begin{subfigure}[t]{0.48\textwidth}
        \centering
        \includegraphics[width=\linewidth]{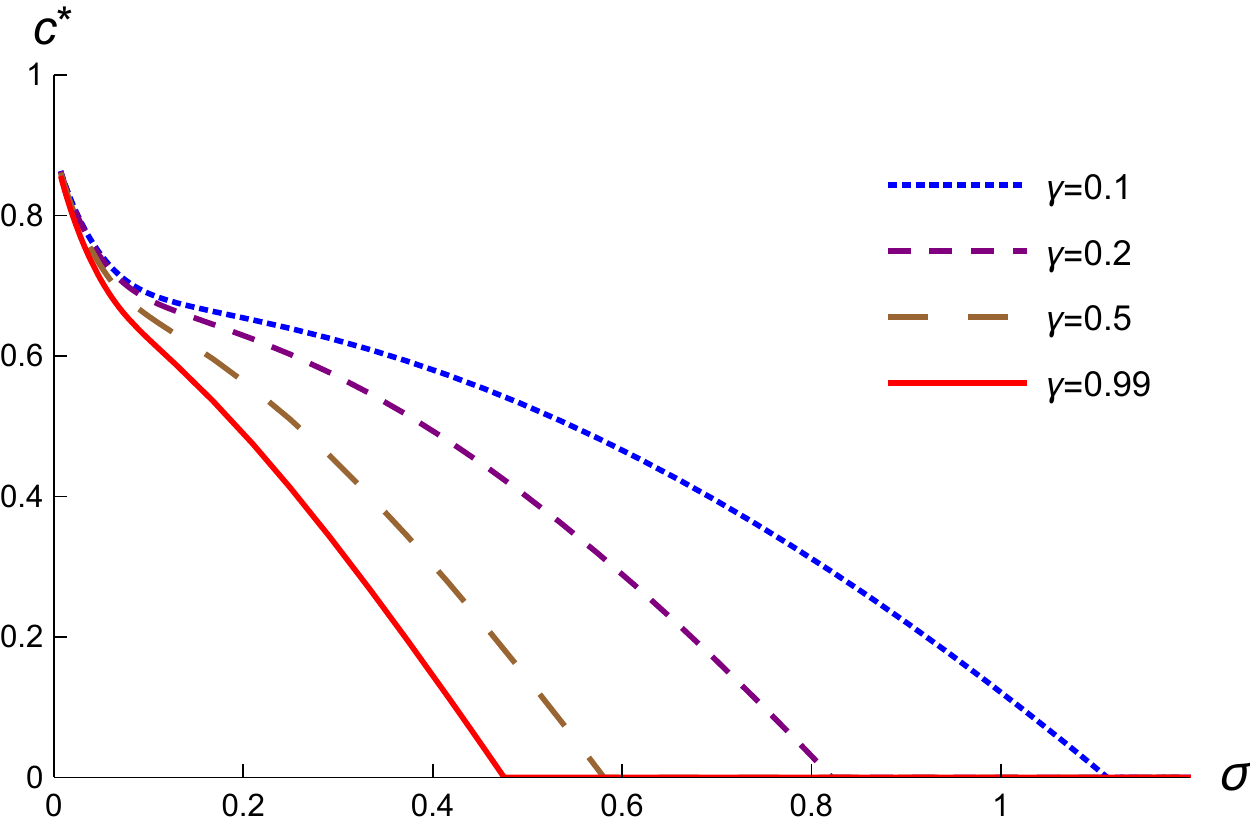} 
        \caption{$\mu=0.45$.} \label{fig:statics2}
    \end{subfigure}
	\caption{Optimal interval thresholds for Normal distributions (mean $\mu$, variance $\sigma^2$) and linear-quadratic Proposer utility, $u(a)=-(1-\gamma)|1-a|-\gamma(1-a)^2$.}
	\label{fig:compstats}
\end{figure}

While we do not have general comparative statics results in the variability of the type distribution, 
%we could add that we also verified it for other distributions (e.g., logistic -- the Mathematica file contains a figure for that too)
it can be shown that for any strictly unimodal distribution with mode $\mode\geq 0$, an optimal interval delegation set has more compromise than if Proposer knew Vetoer's type to be $\mode$. That is, if $c^*$ is an optimal interval threshold, then $c^*\leq \min\{2 \mode,1\}$; the inequality is strict if $\mode\in(0,1/2)$. Adding this kind of uncertainty to the complete-information model of \citet{RR:78} thus increases the extent of Proposer's compromise.

Comparative statics with respect to changes in the status quo are ambiguous. In particular, one may conjecture that increasing the status quo from $0$ towards $1$ would reduce the optimal amount of discretion or extent of compromise (i.e., increase $C^*$, as in \autoref{prop:compstats} \ref{p:LR}). This is not assured, however, even in the simplest case of linear loss utility and a strictly logconcave density. To see that, observe using the subgradient condition discussed after \autoref{prop:no_compromise} that no compromise can be (uniquely) optimal given our status quo of $0$ for suitable logconcave distributions with mode in $(0,1)$.
But if the status quo is raised to some $s$ in between the mode and $1$, then full delegation under the new status quo (i.e., $[s,1]$) becomes optimal. Discretion has increased.

\subsection{Comparisons}
\label{s:discussion}

This subsection compares the outcome of optimal delegation with two game forms considered in earlier work.

A natural starting point is the incomplete-information version of the \citet{RR:78} model. Proposer makes a take-it-or-leave it proposal $a\in \Reals$, which Vetoer can accept or veto.  This can be viewed as restricting Proposer to singleton delegation sets. Clearly, Proposer is strictly worse off in this institution unless no compromise is the optimal mechanism. We assume throughout this subsection that no compromise is not an optimal interval delegation set; as noted in \autoref{s:nocompromise} it is sufficient that Proposer's utility $u(a)$ is differentiable at his ideal point $a=1$ (hence $u'(1)=0$).\footnote{\label{fn:interior_singleton}A weaker condition suffices: $2u'(1)[1-F(1/2)]<f(1/2)[u(1)-u(0)]$. This ensures that $1$ is not an optimal singleton proposal, nor is $\{1\}$ an optimal interval delegation set. Recall that when $u$ is not differentiable at $1$, $u'(1)$ refers to the left derivative.} We will see below that not only does Proposer strictly benefit from optimal delegation, but so does Vetoer under some conditions, even when full delegation is not optimal for Proposer.

\citet{Matthews:89} studies cheap talk before veto bargaining: prior to Proposer making a singleton proposal, Vetoer can send a costless and nonbinding message. As usual in cheap-talk games there is an uninformative and hence noninfluential equilibrium, in which Proposer makes the same proposal, $a_U >0$, as he would absent the possibility of cheap talk. \citeauthor{Matthews:89} provides conditions under which there is also an equilibrium with informative and influential cheap talk; it is sufficient given the support of our type density that $u'(1)=0$ (or that even the weaker condition in fn.~\ref{fn:interior_singleton} holds). A set of low Vetoer types pool on a ``veto threat'' message, while the complementary set of high types pool on an ``acquiescing'' message. In response to the latter message, Proposer offers $a=1$; in response to the veto threat Proposer offers some $a_{I} \in (0,1)$. The former proposal is accepted by all types that acquiesced, while the latter is accepted by only a subset of types that made the veto threat; types below some strictly positive threshold exercise the veto. An influential cheap-talk equilibrium is outcome equivalent to the delegation set $\{a_I,1\}$ in our framework.

There can be multiple cheap-talk equilibria with distinct outcomes, both among influential equilibria and among noninfluential equilibria (i.e., distinct $a_I$ and $a_U$ respectively). \citeauthor{Matthews:89} shows that $a_I<a_U$ in any two equilibria of the respective kinds; moreover, he provides conditions under which $a_I$ is unique, i.e., all influential cheap-talk equilibria have the same outcome \citep[Remark 3]{Matthews:89}. As elaborated in the proof of our \autoref{p:pareto}, multiplicity is ruled out when %there is a unique solution in $(0,1)$ to 
%\begin{equation}
%\label{e:influential}
%\frac{2u'(a)}{u(a)-u(0)}=\frac{f(a/2)}{F((1+a)/2)-F(a/2)}.
%\end{equation}
the following function has a unique zero:
\begin{equation}
\label{e:influential}
{2u'(a)}\left[{F((1+a)/2)-F(a/2)}\right]-{f(a/2)}{\left[u(a)-u(0)\right]}.
\end{equation}

By way of comparison, we recall that a zero of a similar function given in \eqref{e:opt_del} is the first-order condition for optimality of an interval delegation set's threshold.

\begin{proposition}\label{p:pareto}
Assume no compromise is not an optimal delegation set, and that either \eqref{e:opt_del} or \eqref{e:influential} is strictly downcrossing on $(0,1)$.\footnote{A function $h(a)$ is strictly downcrossing if for any $a_L<a_H$, $h(a_L)\leq 0\implies h(a_H)<0$.} Any optimal interval delegation set $[c^*, 1]$ has $c^*<\min\{a_I,a_U\}$ for any influential and noninfluential cheap-talk equilibrium $a_I$ and $a_U$, respectively. Hence, if $[c^*,1]$ is an optimal delegation set, then it strongly Pareto dominates any cheap-talk outcome, influential or not.
\end{proposition}

By strong Pareto dominance, we mean that Proposer is ex-ante better off, while Vetoer is better off no matter his type; moreover, a set of Vetoer types that have strictly positive probability are strictly better off.
\autoref{p:pareto}'s conclusions hold trivially when full delegation is the optimal delegation set ($c^*=0$), even without its hypotheses. But under its hypotheses, the conclusions also apply to other optimal intervals. The function \eqref{e:opt_del} is strictly downcrossing on $(0,1)$ under \autoref{LQ} and either strict logconcavity of the type density or strict concavity of Proposer's utility. Indeed, this underlies the uniqueness claim in \autoref{c:LQinterval}; see \autoref{l:qconvex}. Moreover, \autoref{c:LQinterval} also assures that interval delegation is then optimal.

%Per \autoref{c:LQinterval}, \autoref{e:opt_del} has a unique solution under quadratic loss %\autoref{LQ} (with $\gamma>0$, since we assume $u$ is not linear in this section)  and a logconcave type density, so long as full delegation is not optimal. Of course, \autoref{p:pareto}'s conclusion holds trivially if full delegation is the optimal delegation set. Therefore, under \autoref{LQ} and a logconcave type density, every optimal interval delegation set has a strictly lower threshold than any proposal in any cheap-talk equilibrium; moreover, recall from \autoref{c:LQinterval} that interval delegation is an optimal mechanism.

Here is the intuition behind \autoref{p:pareto}. Consider Proposer's tradeoff when marginally lowering his proposal $a_I\in (0,1)$ in an influential cheap-talk equilibrium. The benefit is that some types just below $a_I/2$ will accept rather than veto; the cost is that  the action induced from all types in the interval $(a_I/2,(1+a_I)/2)$ is lower. When Proposer instead delegates the interval $[a_I,1]$, the benefit from lowering $a_I$ is unchanged while the cost is reduced, because types above $a_I$ are now unaffected by the change. Proposer is thus more willing to compromise when choosing among interval delegation sets rather than under cheap talk.\footnote{The hypotheses in \autoref{p:pareto} ensure that this local-improvement intuition extends to global optimality.} Consequently, all Vetoer types benefit---at least weakly, and some strictly---from optimal interval delegation as compared to cheap talk. While Proposer could be harmed by a restriction to interval delegation, there is strong Pareto dominance when intervals constitute optimal delegation.

We note that if interval delegation is not optimal, then some Vetoer types may be worse off under optimal delegation than under cheap talk. For example, it is possible that the optimal delegation set takes the form $\{a^*,1\}$ with $a^*\in (0,1)$. In this case one can show that necessarily $a^*<a_I$ in any influential cheap-talk equilibrium; intuitively, while $a_I$ is sequentially rational, committing to a lower proposal helps ex ante by inducing action $1$ rather than $a_I$ from some types. Consequently, while Proposer strictly benefits from optimal delegation, some Vetoer types would strictly prefer either cheap-talk outcome.

%Finally, we note the equilibrium with informative cheap talk is never the optimal mechanism, in contrast to what one may expect given a classic result in expertise based delegation \citep{Dessein:02}, an equilibrium with informative cheap talk may outperform full delegation.\footnote{For a simple example, suppose $v$ is normally distributed with mean $m \in (0, 1/2)$ and variance $\sigma$.  Informative cheap talk is always possible, but as $\sigma \rightarrow 0$ the payoff from full delegation converges to $u(1-m)$ and from a singleton proposal converges to $u(1-2m)>u(1-m)$.}

\section{Applications}
\label{s:applications}

We now discuss some implications and interpretations of our analysis in the context of three applications. % mentioned in the \hyperref[s:intro]{introduction}.

\paragraph{Menus of products.} Our framework can be applied to questions of which products to present customers with, albeit in a stylized manner. For an illustration, suppose a salesperson has at his disposal a set of products indexed by $a\in [0,1]$, with higher $a$ corresponding to higher quality. The price of product $a$ is $k a^2$, where $k>0$. This pricing can be interpreted as emerging from a constant markup on a quadratic cost. Consumers vary in how they trade off quality and price; specifically, a consumer of type $v\geq 0$ has gross valuation $v a$, and hence net-of-price payoff $va-k a^2$. 
If a consumer does not purchase, his payoff is $0$; a consumer cannot purchase a product he is not shown (perhaps because of ignorance, or because the salesperson can claim it is unavailable). Observe that we can normalize $k=1/2$, as this simply rescales the consumer type $v$. The salesperson receives a higher commission on better products, reflected by his strictly increasing and concave utility $u(a)$. 
%The salesperson has a strictly increasing and concave utility $u(a)$; this may reflect the nature of his commissions on better products, or he may simply be concerned with revenue but is risk averse.
Given any belief density the salesperson holds about a particular consumer's valuation the salesperson's problem of which products to show the consumer is precisely that of determining the optimal delegation set in our setting.

Take the case of a linear $u$. Propositions \ref{prop:full}--\ref{prop:interval} imply that if the density of $v$ is logconcave, it is optimal to show the consumer some set of ``best products'' (i.e., an interval of products $[c,1]$); if the density is strictly decreasing, then all products should be shown; and if the density is strictly increasing, then only the highest-quality product should be shown. \autoref{prop:compstats}\ref{p:risk} implies that if the commission schedule changes to make $u$ more concave, the salesperson shows a larger set of products. \autoref{prop:compstats}\ref{p:LR} implies that if wealthier consumers (or, if wealth is unobservable, some proxy thereof) have a higher distribution of $v$ in the likelihood ratio sense, then wealthier consumers are shown a smaller set of products.

What if a consumer can choose the information to disclose about her type?\footnote{\citet{ALV:19}, \citet{HV19}, and \citet{Ichihashi:20} consider optimal consumer disclosure in models that emphasize price discrimination.} Specifically, suppose, as is standard in voluntary disclosure models, that any type $v$ can send any message (a closed subset of $\Reals_+$) that contains $v$. The salesperson decides on the product menu after observing the message. No matter the type distribution, there are at least two equilibria: one in which no type discloses any information, and one in which all types fully disclose.\footnote{For any unused message, $V\subseteq \Reals_+$, let the salesperson put probability $1$ on $v=\max V$ (or, if $\sup V=\infty$, on some $v\in V$ with $v\geq 1$) and offer the correspondingly optimal singleton menu. It is then straightforward that no consumer type does strictly better by deviating to any unused message.} Every consumer type prefers the former equilibrium to the latter; some types have a strict preference unless nondisclosure results in only a single product being shown (\autoref{prop:no_compromise}). In general there can be other equilibria, some of which may dominate the nondisclosure equilibrium in terms of ex-ante consumer welfare.\footnote{For example, suppose the type density is strictly decreasing on a small interval $[0,\delta]$ and strictly increasing thereafter, and $u$ is linear. Then, under nondisclosure, the salesperson's optimal menu is the singleton $\{1\}$ (\autoref{prop:no_compromise}). There is also a partial-disclosure equilibrium in which types $[0,\delta]$ pool on the message $[0,\delta]$ and all higher types pool on the message $[\delta,\infty)$; the former message leads to the menu $[0,\delta]$ by the full-delegation logic of \autoref{c:decr}, while the latter message leads to the singleton $\{1\}$ by the no-compromise logic of \autoref{prop:no_compromise}. Every consumer type prefers this partial-disclosure equilibrium to the nondisclosure equilibrium, some strictly.}  However, when the salesperson offers all products under the prior (\autoref{prop:full}), the nondisclosure equilibrium is consumer optimal---not only ex ante, but for every consumer type.

%\paragraph{Default tip menus.} Our model is quite well suited to addressing what options a firm should present its customers with, albeit in a stylized manner. An interesting case is that of which tip defaults to make available to consumers on mobile apps, or to pre-compute on bills at restaurants, etc. While ``behavioral'' consumer aspects are no doubt relevant, a simple starting point would be to assume that a consumer whose preferred tip is $v$ (percent, say) will choose the default that is closest to $v$. It is commonplace to offer three defaults, in addition to ``no tip''.

\paragraph{Lesser-included offenses.} The legal doctrine of lesser-included offenses in criminal cases is ``the concept that a defendant may be found guilty of an uncharged lesser offense, instead of the offenses formally charged \ldots a recognized and well-established feature of the American criminal justice system'' \citep{Adlestein95}. For instance, ``the lesser-included offenses of first degree murder include second degree murder, voluntary manslaughter, involuntary manslaughter, criminally negligent homicide, and aggravated assault'' \citep{OS08}.
%``provides that a criminal defendant may be convicted at trial of any crime supported by the evidence which is less than, but included within, the offense charged by the prosecution'' \citep{SS:95}.
 The framework studied in our paper provides a formal lens to understand welfare implications of the doctrine for both prosecutors and defendants.

Our model views $v$ as a jury's (or judge's) evaluation of the optimal penalty or true severity of a crime. The defendant is delivered a penalty or sentence corresponding to the most severe charge on which there is a conviction. Verdict $0$ corresponds to a complete acquittal, which is always available to the jury, while $1$ is the maximum penalty or the maximum charge the prosecutor can realistically put forward in a given case. We assume the jury will convict the defendant of the closest charge to $v$ that is available.\footnote{So a jury may convict on an excessive charge if a more appropriate one is not available. The U.S. Supreme Court opined in its ruling on Beck v.~Alabama that ``when the evidence establishes that the defendant is guilty of a serious, violent offense but leaves some doubt as to an element justifying conviction of a capital offense, the failure to give the jury such a `third option' [a lesser offense] inevitably enhances the risk of an unwarranted conviction.'' (447 U.S. 625, 1980).} 
%The defendant can only be convicted on one charge, or is delivered the sentence corresponding to the most severe charge on which there is a conviction. 
 In a stylized manner, the lesser-included offenses doctrine can be modeled as implying that if a charge $a$ is included, then the jury can choose any verdict in $[0,a]$. Plainly, if a prosecutor has the option to put forward any subset of charges---or equivalently, to have the jury selectively instructed about only specific lesser offenses---then the doctrine (or the jury being instructed of it in full) is just a constraint. It necessarily makes the prosecutor worse off ex ante, at least weakly. It follows that if the defendant's utility is the additive inverse of the prosecutor's, then the doctrine can only help the defendant ex ante.

Our analysis clarifies, however, circumstances in which the doctrine does not strictly hurt the prosecutor, or help the defendant, ex ante. Assume the prosecutor's utility $u$ is increasing in the verdict.\footnote{That prosecutors seek to maximize the penalty is a common assumption in law and economics since \citet{Landes71}.} Under the doctrine, the prosecutor then brings the maximum charge. So the prosecutor's ex-ante utility is the same absent the doctrine if and only if full delegation is unconstrained optimal. So long as $u$ is concave, \autoref{prop:full} can be applied; in particular, the doctrine is irrelevant for any prosecutor who is sufficiently risk averse (\autoref{c:risk}). On the other hand, from an ex-post perspective, it is precisely when full delegation is not prosecutor optimal that the prosecutor will strictly benefit and the defendant strictly lose, with positive probability, from the doctrine.

%Conversely, consider the expected utility of a defendant whose utility is strictly decreasing on $[0,1]$. As this is the additive inverse of some prosecutor utility function, the defendant is ex-ante better off  when the prosecutor is constrained by the lesser-included offense doctrine.\footnote{This observation holds regardless of the curvature of the defendant's utility $w(a)$, because a prosecutor with utility $-w(a)$ is ex-ante worse off with a constraint on his delegation set even if, contrary to our maintained assumption, $-w$ is not concave.} But if (and only if) the defendant strictly benefits ex ante, then with positive probability the doctrine strictly harms the defendant ex post: for any delegation set different from full delegation, some defendant types will be convicted of a more serious charge under the doctrine, as the doctrine leads to full delegation. The fact that a defendant can be harmed by the doctrine ex post could explain why defendants sometimes appeal verdicts on the basis that their conviction on a lesser-included offense is not legally valid.\footnote{E.g., in \href{https://law.justia.com/cases/oklahoma/court-of-appeals-criminal/1999/35651.html}{Shrum v.~State (1999)}, the defendant, who was convicted for manslaughter on a charge of murder, appealed that ``1) heat of passion manslaughter is not a `necessarily included offense' of premeditated murder; 2) a jury must acquit when the evidence supports a charge not alleged in the Information''. The appeal was denied.} 

The foregoing discussion assumes the prosecutor can bring any set of charges.  If the prosecutor were restricted to bringing a single charge, then the welfare implications of the lesser-included offenses doctrine are less clear cut. The issue boils down to whether full delegation is ex-ante preferred by the relevant party to the prosecutor's optimal single charge. We observe that now the doctrine ex-ante benefits the prosecutor and hurts the defendant when full delegation is unconstrained optimal, whereas the comparison is reversed when no compromise is unconstrained optimal (\autoref{prop:no_compromise}).

\paragraph{Legislatures and Executives.} 
%Legislatures write bills that can be vetoed by executives. One interpretation of our results is that they predict bills granting the executive---or a bureaucracy that implicitly shares the executive's preferences---more discretion when there is greater preference misalignment between the executive and the legislature, in the sense of \autoref{prop:compstats}\ref{p:LR}. This would reverse the comparative static emphasized in political science by \citet{EOH:99} and others, which stems from expertise-based delegation models.\footnote{For exceptions and caveats, see, for example, \citet{Volden02} and \citet{HM:04}. \citet[p.~112]{Volden02} notes that modeling the executive's veto is important for his finding that ``there are conditions under which discretion is increased upon a divergence in legislative-executive preferences''. The mechanism underlying his findings is different from that in this paper, however; in particular, expertise-based delegation is still essential to his analysis.} %He focuses on how the bureaucracy's preferences compares with the legislature's and the executive's, and also on the degree of discretion already present in the status quo. Our mechanism turns on Proposer's uncertainty about Vetoer's preferences. Of course, in practice one expects both the expertise-based rationale and our veto-based one to coexist to varying degrees.

Legislatures write bills that can be vetoed by executives. But executives do more than just approve or veto: as emphasized by \citet[pp.~378--379]{EOH:96}, ``all laws passed by Congress are implemented by the executive branch in one form or another'', and, since, ``Presidents generally appoint administrators with preferences similar to their own'' the amount of discretion given is a ``key variable in \ldots congressional-executive relations''.  One interpretation of our results is that they predict bills granting the executive \emph{more} discretion when there is greater preference misalignment between the executive and the legislature, in the sense of \autoref{prop:compstats}\ref{p:LR}. This flips the comparative static emphasized in the political science literature \citep{EOH:96, EOH:99}, which stems from expertise-based delegation models.\footnote{For exceptions and caveats, see, for example, \citet{Volden02} and \citet{HM:04}. \citet[p.~112]{Volden02} notes that modeling the executive's veto is important for his finding that ``there are conditions under which discretion is increased upon a divergence in legislative-executive preferences''. The mechanism underlying his findings is different from that in this paper, however; in particular, expertise-based delegation is still essential to his analysis.} Of course, in practice one expects both the expertise-based rationale and our veto-based one to coexist to varying degrees.

There are examples of legislatures apparently providing misaligned executives with veto power greater discretion than the legislature would consider optimal. A case in point is the ``U.S. Troop Readiness, Veterans' Care, Katrina Recovery, and Iraq Accountability Appropriations Act'' (House Resolution 2206) enacted in May 2007 that provided funding for the United States' war in Iraq. Congress passed an earlier version, House Resolution 1591, that set a deadline of April 2008 for U.S. troops to withdraw from Iraq. This suggests that, even after accounting for any expertise-based delegation rationale, Congress preferred a relatively tight deadline. But the bill was vetoed by President George W. Bush. Although Democrats controlled both chambers of Congress, they did not have the requisite supermajority to override the veto. To secure the President's approval, the eventual Act replaced the withdrawal deadline with vague metrics that gave the President more discretion.

%An earlier version, House Resolution 1591, passed Congress but was vetoed by President George W. Bush because it set a deadline of April 2008 for U.S. troops to withdraw from Iraq. 
%Although Democrats controlled both chambers of Congress, they did not have the requisite supermajority to override the veto. To secure President Bush's approval, the eventual Act replaced the withdrawal deadline with vague metrics that gave the President discretion.

We must also stress an alternative perspective on our results: rather than passing bills that grant ex-post discretion, discretion can manifest in the executive effectively selecting which bill (from some subset, none of which grant ex-post discretion) the legislature passes. For example, the President may be consulted by Congress about different versions of legislation. These two forms of discretion are equivalent within our model. An empirical test of our model's predictions in the political arena would have to overcome this challenge, and that of the coexistence of expertise- and veto-based delegation rationales.

\section{Conclusion}
\label{s:conclusion}

We have studied Proposer's optimal mechanism, absent transfers, in a simple model of veto bargaining. Our main results identify sufficient and necessary conditions  for the optimal mechanism to take the form of certain delegation sets, including full delegation, no compromise, and more generally interval delegation. While we have focused on a quadratic loss function for Vetoer, our analytical methodology can be applied to deduce optimality of these delegation sets for a broader class of Vetoer preferences. Specifically, the methods can be readily applied to Vetoer utility functions of the form $v a+b(a)$, for any differentiable and strictly concave function $b$. The conditions in Propositions \ref{prop:full}--\ref{prop:interval} would be more complicated, however. Our methodology could also be used to deduce optimality of other kinds of delegation sets, for example Proposer offering his ideal point and one additional compromise option. %\redt{or offering ``reverse intervals'' of the form $[0,c]\cup \{1\}$ for some $c \in (0,1)$.}

In some applications it is plausible that Vetoer can choose among multiple default options. For instance, there may be two internal candidates available to an organization for an open position even if it rejects those put forward by a search committee. Formally, suppose Vetoer has available a finite set of actions to choose among if she exercises her veto. Proposer's optimal delegation set can be obtained by simply solving a series of separate problems analogous to ours, and then ``stitching'' the solutions together. Let us illustrate assuming Vetoer has only two options upon a veto, which we denote $0$ and $a^*>0$. If $a^*>1$, then Proposer solves the problem we have studied---with a single veto option at $0$---to determine a delegation set $A\subseteq [0,1]$; separately, she solves an analogous problem in which the single veto option is $a^*$ to determine a delegation set $A'\subseteq [1,a^*]$. The overall optimal delegation set is simply $A \union A'$. If, on the other hand, $a^*<1$, then Proposer determines an optimal $A\subseteq [0,a^*]$ with veto option $0$ and her ideal point viewed as $a^*$, and an optimal $A'\subseteq [a^*,1]$ with the veto option viewed as $a^*$ and her true ideal point $1$; the overall optimal delegation set is again $A\union A'$. As an example, assume Proposer has linear loss utility and the type distribution is unimodal with mode less than $1$. Following \autoref{rem:unimodal}, our model's solution is interval delegation with a threshold at some $c^* \in [0,1]$. So long as the additional veto option is $a^*\geq c^*$, it follows that the solution is unchanged: if $a^*\in [c^*,1]$, then $a^*$ was already part of the optimal delegation set; if $a^*>1$, then the decreasing density to the right of $1$ means, by \autoref{prop:no_compromise}, that no compromise is optimal on $[1,a^*]$.

A key assumption underlying our analysis is that of Proposer commitment. In some contexts Proposer may be unable to preclude reapproaching Vetoer with another proposal (or menu of proposals) following a veto. When the optimal mechanism in our setting is full delegation, we believe that such lack of commitment is not problematic. By offering the full-delegation menu to begin with, bargaining will effectively conclude at the first opportunity.

When full delegation is not optimal, however, matters are considerable more nuanced. Sequential veto bargaining without commitment has received only limited theoretical attention, largely in finite-horizon models with particular type distributions \citep[e.g.,][]{Cameron:00}. In ongoing research, we are studying an infinite-horizon model. Our preliminary results suggest that, owing to single-peaked preferences, non-Coasian dynamics can emerge that allow Proposer to obtain her commitment solution when players are patient.

\vspace{.2in}

\bibliographystyle{ecta}
\bibliography{delegation}

\newpage

\appendix
\numberwithin{equation}{section}
\numberwithin{proposition}{section}
\numberwithin{lemma}{section}
\numberwithin{example}{section}

{\Large \noindent \textbf{Appendices}}
\label{s:appendices}

\section{Proofs of Propositions \ref{prop:full}, \ref{prop:no_compromise}, and \ref{prop:interval}}
\label{s:proofs}

In \autoref{s:proofs} we assume the support of the type distribution $F$ is $[0,1]$. This is without loss (even among stochastic mechanisms) because it is always optimal for Proposer to choose action $1$ for types above $1$ and, given the outside option, to choose action $0$ for types below $0$.\footnote{Formally, consider any IC and IR mechanism $m$. %IR implies that for all $v\geq 0$, $\E[m(v)]\geq 0$. 
Define another mechanism $\tilde{m}$ 
such that for any type $v$, $\tilde m(v)$ is a lottery among
\[
%\{\delta_0,\delta_1\} \union \left\{\ell: \E[\ell]\le 1 \text{ and there is } \hat{v}\in[0,1] \text{ with } m(\hat{v})=\ell\right\}
\mathcal C := \{\delta_0,\delta_1\} \union \left\{m(\hat v): \hat v \in [0,1] \text{ and } \E_{m(\hat v)}[a]\leq 1\right\}
\]
that type $v$ likes the most in this set, with ties  broken in Proposer's favor. Plainly, $\tilde m$ satisfies IC and IR. Since $\E_{m(v)}[a]\geq 0$ for all $v\geq 0$ (because $m$ satisfies IR), it follows that for any type $v<0$, $\tilde m(v)=\delta_0$. As IR also implies that for any $v\leq 0$, $\E_{m(v)}[a]\leq 0$, it follows that for any $v\leq 0$, Proposer's expected utility under $\tilde m(v)$ is higher than his expected utility under $m(v)$. For any type $v\ge 1$, $\tilde{m}(v)=\delta_1$ is Proposer's ideal action. For any $v\in(0,1)$, there are two cases. First, for any $v$ such that $\E_{m(v)}[a]> 1$, $\delta_1$ is uniquely optimal for $v$ in $\mathcal C$. Second, for any $v$ such that $\E_{m(v)}[a] \leq 1$, $m$ being IC implies that either $\delta_1$ or $m(v)$ is optimal for $v$ in $\mathcal C$, and we have specified that ties are broken in favor of Proposer. In either case, Proposer's expected utility under $\tilde m(v)$ is higher than his expected utility under $m(v)$ for $v\in (0,1)$. In sum, Proposer prefers $\tilde{m}$ to $m$.}

\subsection{Sufficient Conditions}
\label{s:proofs_full}

For convenience, we recall Proposer's problem \eqref{e:relaxed2}: 
\begin{align}\label{e:relaxed2a}
&\max_{m\in \mathcal{S}} \int_0^1 \E_{m(v)}[u(a)] \mathrm dF(v)\tag{P} \\
&\text{s.t. }\E_{m(v)}\left[av  - {a^2}/{2}\right] - \int_0^{v} \E_{m(x)}[a] \mathrm dx = 0\hspace{.5cm} \forall v\in[0,1].\tag{IC-env}\label{eq:stochastic_envelope_a}
\end{align}
We also recall the relaxed problem \eqref{e:relaxed}:
\begin{align}\label{e:relaxeda}
&\max_{\alpha\in \mathcal{A}} \int_0^1 \left(u(\alpha(v)) - \kappa\left[v \alpha(v) -\frac{\alpha(v)^2}{2} - \int_0^{v} \alpha(x) \mathrm dx\right]\right)\mathrm dF(v)\tag{R} \\
&\text{s.t. } v \alpha(v) - \frac{\alpha(v)^2}{2} - \int_0^{v} \alpha(x)\mathrm dx \ge 0 \hspace{.5cm} \forall v\in[0,1].\notag
\end{align}
Problem \eqref{e:relaxed2a} concerns stochastic mechanisms while problem \eqref{e:relaxeda} concerns deterministic ones. In general, there need be no deterministic mechanism that solves problem \eqref{e:relaxed2a}. In particular, the solution to problem \eqref{e:relaxeda}
need not be incentive compatible (as problem \eqref{e:relaxeda} only has a relaxed  incentive constraint) and hence need not be feasible in problem \eqref{e:relaxed2a}. The example in \autoref{s:stochastic} illustrates. However, the following result holds:

 \begin{lemma}
\label{lem:relaxed-stochastic}
Suppose $\alpha^*\in\mathcal{A}$ solves problem \eqref{e:relaxeda} and is incentive compatible. Then $\alpha^*$ also solves \eqref{e:relaxed2a}.
\end{lemma}

\begin{proof}
To obtain a contradiction, suppose $\alpha^*$ does not solve \eqref{e:relaxed2a}. Since $\alpha^*$ is, by assumption, feasible for \eqref{e:relaxed2a}, there is $m\in \mathcal{S}$ that is feasible for \eqref{e:relaxed2a} and achieves a strictly higher objective value in \eqref{e:relaxed2a} than $\alpha^*$. Define $\overline{\alpha}\in \mathcal{A}$ by setting $\overline{\alpha}(v):= \E_{m(v)}[a]$ for each $v$. It holds that $\int_0^{v} \E_{m(x)}[a] \mathrm dx = \int_0^{v} \overline \alpha(x)\mathrm ds$, while for any $v$ Jensen's inequality implies  $\E_{m(v)}\left[av  - {a^2}/{2}\right]\leq v \overline \alpha(v)-\overline \alpha(v)^2/2$.
Hence, feasibility of $m$ in \eqref{e:relaxed2a} implies feasibility of $\overline \alpha$ in \eqref{e:relaxeda}. Moreover, 
\begin{align*}
&\int_0^1 \left(u(\overline{\alpha}(v)) - \kappa\left[v \overline{\alpha}(v) - \frac{\overline{\alpha}(v)^2}{2} - \int_0^{v} \overline{\alpha}(x) \mathrm  dx\right]\right)\mathrm  dF(v) \\
\ge &\int_0^1 \left( \E_{m(v)}\left[u(a) - \kappa\left(v a -\frac{a^2}{2}\right)\right] + \kappa \int_0^{v}\E_{m(x)}[a] \mathrm  dx \right) \mathrm  dF(v) \\
=& \int_0^1 \E_{m(v)}[u(a)] \mathrm dF(v)\\
> &\int_0^1 u(\alpha^*(v)) \mathrm  dF(v)\\
= &\int_0^1 \left(u(\alpha^*(v)) - \kappa\left[v \alpha^*(v) - \frac{\alpha^*(v)^2}{2} - \int_0^{v} \alpha^*(x) \mathrm  dx\right]\right)\mathrm  dF(v),
\end{align*}
where the first inequality holds because the first line is a concave functional (by the definition of 
$\kappa\equiv \inf_{a\in [0,1)} -u^{\prime \prime} (a)$), the first equality holds because $m$ is feasible in \eqref{e:relaxed2a}, the second inequality holds because of our assumption that $m$ achieves a strictly higher value than $\alpha^*$, and the final equality holds because $\alpha^*$ being IC implies it is feasible in \eqref{e:relaxed2a}. Therefore, $\alpha^*$ is not optimal in \eqref{e:relaxeda}, a contradiction.
\end{proof}

To show that a given delegation set solves the relaxed problem, we construct Lagrange multipliers and show that the induced action rule maximizes the following Lagrangian.
Given $\alpha\in \mathcal{A}$ and an increasing and right-continuous function $\Lambda(v)$, let the Lagrangian be given by
\begin{align}
\L(\alpha,\Lambda) &:= 
\int_0^1 \left( u(\alpha(v)) f(v) - \kappa f(v) \left[ v \alpha(v) -\frac{\alpha(v)^2}{2} - \int_0^{v} \alpha(x)\mathrm dx \right]\right)\mathrm d v \notag\\
&\hspace{.3cm} + \int_0^1 \left( v \alpha(v) -\frac{\alpha(v)^2}{2} - \int_0^{v} \alpha(x)\mathrm dx \right) \mathrm d \Lambda(v)\nonumber\\
&= \int_0^1 \left( u(\alpha(v)) f(v) - \alpha(v)[\kappa F(v) - \Lambda(v)]- \kappa f(v)\left[v \alpha(v) -\frac{\alpha(v)^2}{2}\right]\right)\mathrm d v\nonumber\\
&\hspace{.3cm}+ \int_0^1 \left( v \alpha(v) -\frac{\alpha(v)^2}{2}\right)\mathrm d\Lambda(v)+ \int_0^1 \alpha(v)\mathrm dv [\kappa F(1)- \Lambda(1)],\label{e:Lagrange}
\end{align}
where the second equality follows from integration by parts.\footnote{$\int_0^v \alpha(s)\mathrm ds$ is continuous and $\kappa F(v)- \Lambda(v)$ has bounded variation as the difference of two increasing functions. Hence, the Riemann-Stieltjes integral $\int_0^1 \int_0^v \alpha(s)\mathrm ds \mathrm d[\kappa F(v)- \Lambda(v)]$ exists and integration by parts is valid.}

\begin{lemma}\label{l:Lagrange}
Let $\alpha^*$ be induced by a delegation set.\footnote{That is, there is some delegation set $A$ such that $\alpha^*(v)$ is an action in $A\cup\{0\}$ that type $v$ prefers the most.} Suppose there is an increasing and right-continuous function $\Lambda$ such that $\L(\alpha^*,\Lambda)\ge \L(\alpha,\Lambda)$ for all $\alpha\in \mathcal{A}$. Then $\alpha^*$ solves problem \eqref{e:relaxeda}.
\end{lemma}

Here is the idea. Since $\alpha^*$ is incentive compatible (as it is induced by a delegation set), it is feasible for the relaxed problem \eqref{e:relaxeda} and satisfies all inequality constraints as equalities. This implies that complementary slackness is satisfied for any Lagrange multiplier $\Lambda$. It follows that $\alpha^*$ solves problem \eqref{e:relaxeda} if it maximizes the Lagrangian functional.

\begin{proof}
Let $\Obj(\alpha)$ denote the value of the objective function in problem \eqref{e:relaxed} with mechanism $\alpha$. Since $\alpha^*$ is incentive compatible, 
\[ v \alpha^*(v) - \frac{\alpha^*(v)^2}{2} - \int_0^v \alpha^*(x) \mathrm dx = 0,  \]
and therefore $\Obj(\alpha^*)=\L(\alpha^*,\Lambda)$. 
For any $\alpha\in\mathcal{A}$,
\[ v \alpha(v) - \frac{\alpha(v)^2}{2} - \int_0^v \alpha(x) \mathrm dx \ge 0;  \]
since $\Lambda$ is non-decreasing, this implies $\L(\alpha,\Lambda)\ge \Obj(\alpha)$. We conclude
\[ \Obj(\alpha^*)=\L(\alpha^*,\Lambda) \ge \L(\alpha,\Lambda)\ge \Obj(\alpha) . \qedhere \]
\end{proof}

The following observation will allow us to establish that the Lagrangian is a concave functional of $\alpha$ if $\Lambda$ is increasing.
\begin{lemma}\label{lemma:concave}
Suppose $K$ is right-continuous and increasing and $h:\R^2\rightarrow \R$ is bounded, measurable, and for each value of its second argument concave in its first argument. Then the function $S: L^{\infty}\rightarrow \R$ defined by $S(\alpha):=\int_0^1 h(\alpha(v),v)\mathrm dK(v)$ is concave.
\end{lemma}
\begin{proof}
Fix $\alpha_1, \alpha_2\in L^{\infty}$, $c\in(0,1)$ and let $\alpha_c= c \alpha_1 + (1-c) \alpha_2$. Then 
\begin{align*}
S(\alpha_c) - c S(\alpha_1)- (1-c)S(\alpha_2) = \int_0^1 \big( h(\alpha_c(v),v) - c h(\alpha_1(v),v) - (1-c) h(\alpha_2(v),v)\big)\mathrm dK(v) \ge 0
\end{align*}
because concavity of $h$ implies that the integrand is positive for each $v$ and because $K$ is increasing.
\end{proof}
Note that for each $v$, $u(\alpha(v)) f(v) + \kappa f(v)\frac{\alpha(v)^2}{2}$ is concave in $\alpha(v)$  since its second derivative is given by $f(v)[u''(\alpha(v))+ \kappa]$, which is negative by definition of $\kappa$. This implies that, for each $v$, each integrand in \eqref{e:Lagrange} is a concave function of $\alpha(v)$. Since $\Lambda$ is increasing, \autoref{lemma:concave} implies that the Lagrangian is concave in $\alpha$.

This implies that the Lagrangian is maximized at $\alpha$ if the Gateaux differential satisfies $\partial \L(\alpha,\overline{\alpha}-\alpha,\Lambda)\le 0$ for all $\overline{\alpha}\in \mathcal{A}$. 

If $\Lambda(1)=\kappa F(1)$, \eqref{e:Lagrange} simplifies to
\begin{align*}
\L(\alpha,\Lambda)&= \int_0^1 \left( u(\alpha(v)) f(v) - \alpha(v)[\kappa F(v) - \Lambda(v)]- \kappa f(v)\left[v \alpha(v) -\frac{\alpha(v)^2}{2}\right]\right)\mathrm d v\nonumber\\
&\hspace{.3cm}+ \int_0^1 \left(v \alpha(v) -\frac{\alpha(v)^2}{2}\right)\mathrm d\Lambda(v).
\end{align*}
The Gateaux differential is
{\small\begin{align}
\partial \L(\alpha,\overline{\alpha},\Lambda) &= \int_0^1 \left( \Big[u'(\alpha(v)) f(v) - \kappa F(v) + \Lambda(v)\Big] \overline{\alpha}(v)\right)\mathrm d v + \int_0^1 \left([v - \alpha(v)]\overline{\alpha}(v)\right)\mathrm d [\Lambda(v)- \kappa F(v)] \label{eq:gateaux}\\
&= \int_0^1 \left(\int_v^1 u'(\alpha(x)) f(x) - \kappa F(x) + \Lambda(x)\mathrm dx\right)\mathrm d\overline{\alpha}(v) + \int_0^1 \left(\int_v^1 [x - \alpha(x)]\mathrm d [\Lambda(x)- \kappa F(x)]\right) \mathrm d \overline{\alpha}(v), \label{eq:gateaux2}
\end{align}}
where the second equality obtains using integration by parts.

Below, we construct increasing and right-continuous Lagrange multipliers that satisfy $\Lambda(1)=\kappa F(1)$ such that $\partial \L(\alpha,\overline{\alpha}-\alpha,\Lambda)\le 0$.

We begin with the optimality of full delegation.

\begin{proof}[Proof of the sufficiency part of \autoref{prop:full}]
Note that the action function induced by full delegation is $\alpha(v)=v$. We claim that $\alpha$ maximizes the Lagrangian for the multiplier 
$\Lambda(v)=\kappa F(v)- u'(v) f(v)$ for $v<1$ and $\Lambda(1)=\kappa F(1)$. Note that the multiplier is increasing since $\kappa F(v)- u'(v) f(v)$ is increasing by assumption and $u'(v)\ge 0$. The Lagrangian is therefore maximized at $\alpha$ if $\partial \L(\alpha,\overline{\alpha}-\alpha,\Lambda)\le 0$ for all $\overline{\alpha}\in \mathcal{A}$. Note that the integrand of the first integral in \eqref{eq:gateaux} is 0 for almost every $v$ by choice of $\Lambda$ and the second integral is 0 since $\alpha(v)=v$.
\end{proof}

We next consider the optimality of {no compromise}.  

\begin{proof}[Proof of the sufficiency part of \autoref{prop:no_compromise}]
Note that the action rule induced by no compromise satisfies $\alpha(v)=0$ for $v\in [0,\frac{1}{2})$ and $\alpha(v)=1$ for $v \in [\frac{1}{2},1]$. 
Now suppose for all $s\in [0,1/2)$ and $t\in (1/2,1]$ we have
\[
(u'(1)+\kappa (1-t))\frac{F(t)- F(1/2)}{t-1/2} \geq (u'(0)-\kappa s)\frac{F(1/2)-F(s)}{1/2- s}.
\]
and let $\psi:=\inf_{t\in(1/2,1]}(u'(1)+\kappa (1-t))\frac{F(t)- F(1/2)}{t-1/2}$. Define $\Lambda(v)=\kappa F(1/2)- \psi$ for $v\in[0,1)$ and $\Lambda(1)=\kappa F(1)$.

Let $s\in(1/2,1]$. Note that integration by parts implies $\int_s^{1/2}v \mathrm dF(v)=1/2F(1/2)-sF(s)-\int_s^{1/2} F(v)\mathrm dv$. Since $\Lambda(v)$ is constant on $[0,1)$, the definition of $\psi$ implies that, for any $s\in[0,1/2)$,
\begin{align}
&\int_s^{1/2} u'(\alpha(v)) f(v) - \kappa F(v) + \Lambda(v)\mathrm dv + \int_s^{1/2} v \mathrm d [\Lambda(v)- \kappa F(v)] \nonumber\\
=& u'(0) [F(1/2)-F(s)]+ 1/2 [\Lambda(1/2)- \kappa F(1/2)]- s [\Lambda(s)- \kappa F(s)]\nonumber\\
= &[u'(0)+ \kappa s][F(1/2)-F(s)] -(1/2-s)\psi\le 0.\label{e:no_compromise_1}
\end{align}
Similarly, for any $t\in(1/2,1]$,
\begin{align}
&\int_{1/2}^{t} u'(\alpha(v))f(v) -\kappa F(v) + \Lambda(v) \mathrm dv + \int_{1/2}^{t} [v- \alpha(v)]\mathrm d[\Lambda(v)- \kappa F(v)]\nonumber\\
=& u'(1)[F(t)-F(1/2)] + (t-1)[\Lambda(t)- \kappa F(t)] + 1/2[\Lambda(1/2)- \kappa F(1/2)] \nonumber\\
=& [u'(1)+\kappa(1-t)][F(t)-F(1/2)] - (t-1/2)\psi \ge 0.\label{e:no_compromise_2}
\end{align}

Fix arbitrary $\overline{\alpha}\in \mathcal{A}$ that satisfies $\overline{\alpha}(1)=1$. It follows from \eqref{eq:gateaux2} and the definition of $\alpha$ that
\begin{align*}
&\partial \L(\alpha,\overline{\alpha}- \alpha,\Lambda)\\
=&\int_0^1 \left[\int_v^1 (u'(\alpha(x)) f(x) - \kappa F(x) + \Lambda(x))\mathrm dx+ \int_v^1 ([x - \alpha(x)]\mathrm d [\Lambda(x)- \kappa F(x)])\right] \mathrm d [\overline{\alpha}(v)-\alpha(v)]\\
 =&\int_0^1 \left[\int_v^{1/2} (u'(\alpha(x)) f(x) - \kappa F(x) + \Lambda(x))\mathrm dx+ \int_v^{1/2} ([x - \alpha(x)]\mathrm d [\Lambda(x)- \kappa F(x)])\right] \mathrm d \overline{\alpha}(v).
\end{align*}
Since $\overline{\alpha}$ is increasing, \eqref{e:no_compromise_1} and \eqref{e:no_compromise_2} imply that $\partial \L(\alpha,\overline{\alpha}- \alpha,\Lambda)\le 0$. Since the optimal action rule chooses action 1 for type 1, we conclude that $\alpha$ is optimal.
\end{proof}

Lastly, we consider optimality of interval delegation. 

\begin{proof}[Proof of the sufficiency part of \autoref{prop:interval}]
The induced action function is $\alpha(v)=0$ for $v< c^*/2$, $\alpha(v)=c^*$ for $c^*/2\le v\le c^*$ and $\alpha(v)=v$ for $v>c^*$.

We propose the following multiplier:
\begin{align*}
\Lambda(v) = \begin{cases}
\kappa F(c^*/2) - u'(c^*) \frac{F(c^*)-F(c^*/2)}{c^*-c^*/2}  &\text{ if } v < c^*\\
\kappa F(v)- u'(v) f(v) &\text{ if } c^*\le v<1\\
\kappa F(1) &\text{ if } v=1.
 \end{cases}
\end{align*}
$\Lambda$ is constant on $[0,c^*)$ and it follows from \autoref{prop:interval}'s condition \ref{interval1} that $\Lambda$ is increasing on $(c^*,1]$. To see that $\Lambda$ is increasing at $c^*$, note that condition \ref{interval2} holds as an equality for $t=c^*$ and hence the derivative of the LHS of  condition \ref{interval2} with respect to $t$ must be negative at $t=c^*$, which yields
\[ -\kappa \frac{F(c^*)-F(c^*/2)}{c^*/2}+u'(c^*)\frac{f(c^*)c^*/2- (F(c^*)-F(c^*/2))}{(c^*/2)^2}\le 0. \]
Hence, $\Lambda$ is increasing at $c^*$. It is thus sufficient to show $\partial \L(\alpha,\overline{\alpha}-\alpha,\Lambda)\le0$ for all $\overline{\alpha}\in \mathcal{A}$.

Note that, for $v\in[c^*,1]$, $v- \alpha(v)=0$ and the definition of $\Lambda$ implies that for $v\in [c^*,1)$,
\begin{align*}
u'(\alpha(v)) f(v) - \kappa F(v) + \Lambda(v) =0.
\end{align*}
Therefore,
\begin{align*}
\partial \L(\alpha,\overline{\alpha}-\alpha, \Lambda) &= \int_0^{c^*} \left(\int_v^{c^*} u'(\alpha(x)) f(x) - \kappa F(x) + \Lambda(x)\mathrm ds\right) \mathrm d[\overline{\alpha}(v)- \alpha(v)] \\
&+ \int_0^{c^*} \left(\int_v^{c^*} [x - \alpha(x)]\mathrm d [\Lambda(x)- \kappa F(x)]\right) \mathrm d [\overline{\alpha}(v)- \alpha(v)]
\end{align*}
Since $\alpha$ is constant on $[0,c^*/2)$ and $[c^*/2,c^*]$, $\overline{\alpha}-\alpha$ is increasing on $[0,c^*/2)$ and $[c^*/2,c^*]$.
Hence, the following conditions are sufficient for $\alpha$ to maximize the Lagrangian:
$$\int_{t}^{c^*} \big(u'(c^*) f(v) - [\kappa F(v) - \Lambda(v)]\big) dv + \int_{t}^{c^*} (v-c^*) d[\Lambda(v)-\kappa F(v)] \le 0$$
 for $t\in[c^*/2,c^*]$, with equality at $t=c^*/2$, and
$$\int_{s}^{c^*/2} \big(u'(0) f(v) - [\kappa F(v) - \Lambda(v)]\big)\mathrm dv + \int_{s}^{c^*/2} v\mathrm d[\Lambda(v)-\kappa F(v)] \le 0$$
for $s\in [0,c^*/2)$.

Note that $\int_{t}^{c^*} \big(F(v) + (v-c^*) f(v)\big)\mathrm dv = (c^*- t)F(t)$ and $\Lambda$ is constant on $[0,c^*)$. Hence, using the definition of $\Lambda$, we get that for $t\in [c^*/2,c^*]$, 
\begin{align*}
& \int_{t}^{c^*} \left(u'(c^*) f(v) - [\kappa F(v) - \Lambda(v)]\right)\mathrm dx + \int_{t}^{c^*} (v-c^*)\mathrm d[\Lambda(v)-\kappa F(v)]\\
=&\int_{t}^{c^*} \left(u'(c^*) f(v) - \left[\kappa F(v) - \kappa F(c^*/2) + \frac{1}{c^*-c^*/2}\int_{c^*/2}^{c^*} u'(c^*)\mathrm dF(x)\right]\right)\mathrm dv - \kappa \int_{t}^{c^*} (v-c^*)\mathrm dF(v)\\
=&u'(c^*)[F(c^*)-F(t)]  - u'(c^*)[F(c^*)-F(c^*/2)]\frac{c^*-t}{c^*-c^*/2} + \kappa (c^*-t) [F(c^*/2)-F(t)]\\
=&-[u'(c^*)+\kappa (c^*-t)][F(t)-F(c^*/2)]  + (t-c^*/2) u'(c^*)\frac{F(c^*)-F(c^*/2)}{c^*-c^*/2} \\
\le &0,
\end{align*}
where the inequality is by \autoref{prop:interval}'s condition \ref{interval2}, and holds with equality for $t=c^*/2$.

Analogously, note that $\int_s^{c^*/2} F(v) + v f(v) \mathrm dv = c^*/2F(c^*/2)- sF(s)$ and $\Lambda$ is constant on $[0,c^*/2]$. Hence, for $s\in [0,c^*/2]$, 
\begin{align*}
&\int_{s}^{c^*/2} \left(u'(0) f(v) - [\kappa F(v) - \Lambda(v)]\right)\mathrm dv + \int_{s}^{c^*/2} v\mathrm d[\Lambda(v)-\kappa F(v)]\\
=&u'(0) [F(c^*/2)-F(s)] + \Lambda(s)(c^*/2-s) - \kappa [ c^*/2F(c^*/2)- sF(s)] \\
=&[u'(0)-\kappa s] [F(c^*/2)-F(s)] -  u'(c^*) \frac{F(c^*)-F(c^*/2)}{c^*-c^*/2}(c^*/2-s)  \\
\le &0,
\end{align*}
where the inequality is by \autoref{prop:interval}'s condition \ref{interval3}.
Hence, $\alpha$ is optimal.
\end{proof}

\subsection{Necessary Conditions}
\label{s:nondeterministic}

\renewcommand{\t}{\tilde{t}}
\begin{lemma}\label{l:necessity_linear}
Suppose \autoref{LQ} holds. If $\alpha^*$ is deterministic and solves problem \eqref{e:relaxed2a} then it also solves problem \eqref{e:relaxeda}.
\end{lemma}

\begin{proof}
The proof is by contraposition: assuming there exists $\alpha\in \mathcal{A}$ that is feasible for \eqref{e:relaxeda} and achieves a strictly higher objective value in \eqref{e:relaxeda} than $\alpha^*$, we will construct a solution to \eqref{e:relaxed2a} that achieves a strictly higher objective value than $\alpha^*$.

\medskip

\textbf{Claim 1}: There exists $\tilde{\alpha}\in \mathcal{A}$ that is feasible for \eqref{e:relaxeda}, satisfies $\tilde{\alpha}(v)\le 1$ for all $v$, $v \tilde{\alpha}(v) -\frac{\tilde{\alpha}(v)^2}{2} - \int_0^v \tilde{\alpha}(s)\mathrm  ds =0$ for all $v$ such that $\tilde{\alpha}(v)=1$, and achieves a  weakly higher objective value in problem \eqref{e:relaxeda} than $\alpha$.

\medskip

We can assume $\alpha(v)\le 1$ since $u$ is decreasing above $1$. Now suppose instead that $v \alpha(v) -\frac{\alpha(v)^2}{2} - \int_0^v \alpha(s)\mathrm  ds >0$ for some $v$ such that $\alpha(v)=1$. Consider an auxiliary setting in which a principal chooses a pair of functions $(\alpha,t)$ and an agent with type $v$ gets utility $v \alpha(v)-\frac{\alpha(v)^2}{2}-t(v)$. Since $\alpha$ is monotonic, it follows from standard arguments that there exist transfers $t:[0,1]\rightarrow \R$ such that $(\alpha,t)$ is incentive compatible in the auxiliary setting \citep[e.g.,][]{AB:13}. For all $v$, these transfers satisfy $t(v)-t(0)=v \alpha(v)-\frac{\alpha(v)^2}{2} - \int_0^v \alpha(s)\mathrm  ds \ge 0$, where the inequality holds because $\alpha$ is feasible for \eqref{e:relaxeda}.
Define $(\tilde{\alpha},\tilde{t})$ by setting $(\tilde{\alpha}(v),\tilde{t}(v))=(\alpha(v),t(v))$ or $(\tilde{\alpha}(v),\tilde{t}(v))=(1,t(0))$, whichever gives an agent with type $v$ higher expected utility (and choosing the latter if type $v$ is indifferent). Note that $\tilde{t}(0)=t(0)$, which together with $t(v)\ge t(0)$ implies $\tilde{t}(v)-\tilde{t}(0)\ge 0$, with equality for any $v$ such that $\tilde{\alpha}(v)=1$.

Observe that $(\tilde{\alpha},\tilde{t})$ corresponds to an incentive compatible direct mechanism: indeed, if type $v$ strictly prefers $(\tilde{\alpha}(v'),\tilde{t}(v'))$ to $(\tilde{\alpha}(v),\tilde{t}(v))$ then $v$ also strictly prefers $(\alpha(v'),t(v'))$ to $(\alpha(v),t(v))$, contradicting the assumption that $(\alpha,t)$ is incentive compatible. 
It follows from the standard characterization of incentive compatible mechanisms that $\tilde{\alpha}$ is increasing, and
\[v \tilde{\alpha}(v) -\frac{\tilde{\alpha}(v)^2}{2} - \int_0^v \tilde{\alpha}(s)\mathrm  ds = \tilde{t}(v)-\tilde{t}(0)\ge 0, \]
with the inequality holding as equality for $v$ such that $\tilde{\alpha}(v)=1$.

Finally, note that $\alpha(v)\le \tilde{\alpha}(v)\le 1$ for all $v$. Also, $\tilde{t}(v)-\tilde{t}(0)\le t(v)-t(0)$, which implies 
\[ v \tilde{\alpha}(v) -\frac{\tilde{\alpha}(v)^2}{2} - \int_0^v \tilde{\alpha}(s)\mathrm  ds \le v \alpha(v) -\frac{\alpha(v)^2}{2} - \int_0^v \alpha(s)\mathrm  ds. \]
It follows that $\tilde{\alpha}$ achieves a weakly higher objective value in problem \eqref{e:relaxeda}. \marker

\medskip

\textbf{Claim 2}:
Let $\tilde{\alpha} \in \mathcal{A}$ be feasible for \eqref{e:relaxeda} and satisfy $\tilde{\alpha}(v)\le 1$  and $v \tilde{\alpha}(v)-\tilde{\alpha}(v)^2/2 - \int_0^v \tilde{\alpha}(s) \mathrm ds =0$ for all $v$ such that $\tilde{\alpha}(v)=1$. There is a stochastic mechanism $m$ such that, for all $v$, $\Prob_{m(v)}(a \le 1)=1$, $\E_{m(v)}[a] = \tilde{\alpha}(v)$, and $\E_{m(v)}\left[v a -\frac{a^2}{2}\right] -\int_0^v \E_{m(s)}[a]  \mathrm ds = 0 $.

\medskip

Intuitively, this is because Vetoer's utility function is quadratic and we can use noise as a substitute for transfers. We provide an explicit construction of the mechanism $m$ below.

For any $v$ such that $\tilde{\alpha}(v)=1$, define $m(v)$ to put mass 1 on action 1. Now fix arbitrary $v$ such that $\tilde{\alpha}(v)<1$ and arbitrary $d\in (-\infty,0]$ and let $t_1(d)=\frac{1- \tilde{\alpha}(v)}{1-d}$. Then $t_1(d)d + (1-t_1(d))1=\tilde{\alpha}(v)$ for all $d$. Moreover, for any $r\in \Reals$ we can choose $d\in(-\infty,0]$ small enough such that
\[ -\frac{1-\tilde{\alpha}(v)}{1-d} d^2- \left( 1- \frac{1-\tilde{\alpha}(v)}{1-d} \right) \le r \]
because the LHS $\to -\infty$ as $d\rightarrow -\infty$. Hence, by choosing $d$ small enough we get
\[v \tilde{\alpha}(v)-t_1(d)\frac{d^2}{2}-(1-t_1(d))\frac{1}{2}-\int_0^v \tilde{\alpha}(s)\mathrm  ds\le 0.\]
Given $v\in[0,1]$ and $t_2\in [0,1]$, we define $m(v)$ to put probability $t_2$ on action $\tilde{\alpha}(v)$, probability $(1-t_2)t_1(d)$ on action $d$, and probability $(1-t_2)(1-t_1(d))$ on action $1$. 
It follows from the above that $m(v)$ satisfies $\E_{m(v)}[a]=\tilde{\alpha}(v)$, $\Prob_{m(v)}(a\le 1)=1$, and we can choose $t_2\in [0,1]$ such that
\[\E_{m(v)}\left[v a -\frac{a^2}{2}\right] -\int_0^v \E_{m(s)}[a]  \mathrm ds = 0.\]
Defining $m(v)$ in this way for all $v$ such that $\tilde{\alpha}(v)<1$, the claim follows. \marker

\medskip

We conclude that $m$ is feasible for \eqref{e:relaxed2a}. Therefore,
\begin{align}
\int_0^1 u(\alpha^*(v))\mathrm dF(v)=&\int_0^1 \left(u(\alpha^*(v)) - \kappa\left[ v \alpha^*(v)-\frac{\alpha^*(v)^2}{2} - \int_0^v \alpha^*(s)\mathrm ds \right] \right)\mathrm dF(v) \nonumber\\
<&\int_0^1 \left(u(\alpha(v)) - \kappa\left[ v \alpha(v)-\frac{\alpha(v)^2}{2} - \int_0^v \alpha(s)\mathrm ds \right] \right)\mathrm dF(v) \nonumber\\
\le & \int_0^1 \left(u(\tilde{\alpha}(v)) - \kappa\left[ v \tilde{\alpha}(v)- \frac{\tilde{\alpha}(v)^2}{2} - \int_0^v \tilde{\alpha}(s)\mathrm ds \right]\right) \mathrm dF(v) \label{e:ineq}
\end{align}
where the equality holds because $\alpha^*$ is feasible for \eqref{e:relaxed2a}, the first inequality holds because we assume that $\alpha$ achieves a strictly higher value than $\alpha^*$, and the second inequality holds by Claim 1.

Under \autoref{LQ}, 
$\kappa\equiv \inf_{v\in[0,1)} -u^{\prime \prime}(v)=2\gamma$. Hence, for any $a,b\le 1$ and $\lambda\in[0,1]$, some algebra shows that
\[u\left(\lambda a+(1- \lambda)b\right)+\kappa \frac{\left[\lambda a+(1- \lambda)b\right]^2}{2} = \lambda\left[u(a)+\kappa \frac{a^2}{2}\right]+(1- \lambda)\left[u(b)+\kappa \frac{b^2}{2}\right].\]
 Since $\Prob_{m(v)}(a\le 1)=1$ and $\E_{m(v)}[a]=\tilde{\alpha}(v)$ for all $v$, 
expression \eqref{e:ineq} therefore equals
\begin{align*}
\int_0^1 \left(\E_{m(v)}\left[u(a) - \kappa\left( v a - \frac{a^2}{2}\right) \right]+ \kappa\int_0^v \E_{m(s)}[a] \mathrm ds\right) \mathrm dF(v).
\end{align*}
Since $m$ is feasible for \eqref{e:relaxed2a}, this expression equals $\int_0^1 \E_{m(v)}[u(a)] \mathrm dF(v)$. 
This contradicts the assumption that $\alpha^*$ solves \eqref{e:relaxed2a}, and we conclude that $\alpha^*$ solves \eqref{e:relaxeda}.
\end{proof}

Let $\phi$ denote the objective function in \eqref{e:relaxeda}. The set of feasible solutions for \eqref{e:relaxeda} is convex, and optimality of $\alpha$ therefore implies $\partial \phi(\alpha,\overline{\alpha}-\alpha)\le 0$ for any $\overline{\alpha}\in \mathcal{A}$ that is feasible for \eqref{e:relaxeda}. Recall the assumption $F(1)=1$.
\begin{align}
 \phi(\alpha) &= \int_0^1 \left(u(\alpha(v)) - \kappa\left[v \alpha(v) -\frac{\alpha(v)^2}{2} - \int_0^{v} \alpha(s)\mathrm  ds\right]\right)\mathrm  dF(v),\nonumber\\
 \partial \phi(\alpha,\overline{\alpha}-\alpha)&= \int_0^1 \left([u'(\alpha(v)) - \kappa[v  -\alpha(v)]](\overline{\alpha}(v)-\alpha(v)) + \kappa \int_0^{v} \overline{\alpha}(s)-\alpha(s)\mathrm  ds\right) \mathrm  dF(v)\nonumber\\
 &= \int_0^1 \left[u'(\alpha(v)) - \kappa\left[v  -\alpha(v)-\frac{1-F(v)}{f(v)}\right]\right](\overline{\alpha}(v)-\alpha(v))\mathrm  dF(v).\label{eq:gateaux_nes1}
\end{align}

\begin{lemma}\label{lemma:interval_nes}
Suppose \autoref{LQ} holds. If a delegation set containing the interval $[a,b]\subseteq [0,1]$ is optimal, then $\kappa F(v)-u'(v)f(v)$ is increasing on $[a,b]$.
\end{lemma}

\begin{proof}
Suppose a delegation set containing the interval $[a,b]$ is optimal, and let $\alpha$ denote the corresponding allocation rule.
Suppose to the contrary that $\kappa F(v)-u'(v)f(v)$ is strictly decreasing for some $v\in[a,b]$. Then $u'(v)f(v) - \kappa f(v)[-\frac{1-F(v)}{f(v)}]$ is strictly increasing on some interval $[d,e]\subset [a,b]$ with $d\neq e$ that contains $v$. 

Set $\overline{\alpha}(v)=\alpha(v)$ for $v\notin [d,e]$ and $\overline{\alpha}(v)=d$ for $v\in[d,e]$. Then $\overline{\alpha}\in \mathcal{A}$ and it is feasible for \eqref{e:relaxeda}. Since $\alpha(v)=v$ for $v\in [d,e]$, it follows from \eqref{eq:gateaux_nes1} that $\partial \phi(\alpha,\overline{\alpha}-\alpha)>0$, a contradiction.
\end{proof}

\begin{proof}[Proof of the necessity part of \autoref{prop:full}]
Suppose \autoref{LQ} holds. The result follows directly from \autoref{lemma:interval_nes}.
\end{proof}

\begin{proof}[Proof of the necessity part of \autoref{prop:no_compromise}]
Suppose \autoref{LQ} holds and let $\alpha\in\mathcal{A}$ be the action rule induced by the delegation set $
\{0,1\}$. Fix $s\in[0,1/2)$ and $t\in(1/2,1]$, $\varepsilon \in(0,1)$ and define
\begin{align*}
\overline{\alpha}_{\varepsilon}(v):= \begin{cases}
0 & \text{ if }v\in[0,s)\\
\varepsilon & \text{ if } v\in [s,1/2)\\
1-\frac{1/2-s}{t-1/2}\varepsilon & \text{ if } v\in [1/2,t)\\
1 & \text{ if } v\in [t,1].
\end{cases}
\end{align*}
Note that
\begin{align*}
 \int_s^{1/2} \left[v-\frac{1-F(v)}{f(v)}\right] \mathrm d F(\theta)&=s[1-F(s)]-1/2[1-F(1/2)], \text{ and }\\
 \int_{1/2}^t \left[v-1-\frac{1-F(v)}{f(v)}\right] \mathrm d F(\theta)&= F(t)(t-1)-F(1/2)(1/2-1)-t+1/2.
\end{align*}
It follows that 
\begin{align*}
\partial \phi(\alpha,\overline{\alpha}_{\varepsilon}- \alpha) = &\varepsilon \int_s^{1/2} \left(u'(0) - \kappa\left[v-\frac{1-F(v)}{f(v)}\right]\right)\mathrm dF(v)\\& -\frac{1/2-s}{t-1/2}\varepsilon \int_{1/2}^{t} \left(u'(1)- \kappa \left[ v-1-\frac{1-F(v)}{f(v)} \right]\right)\mathrm d F(v)\\
= &\varepsilon (1/2-s) \left([u'(0)-\kappa s]\frac{F(1/2)-F(s)}{1/2-s}+\kappa[1-F(1/2)]\right)\\
&-\varepsilon\frac{1/2-s}{t-1/2} [u'(1)+\kappa(1-t)] \left(F(t)-F(1/2)+ \kappa [(1/2-t)F(1/2)+t-1/2]\right)\\
= &\varepsilon (1/2-s) \left([u'(0)-\kappa s]\frac{F(1/2)-F(s)}{1/2-s} - [u'(1)+\kappa(1-t)] \frac{F(t)-F(1/2)}{t-1/2}\right).
\end{align*}
Therefore, $\partial \phi(\alpha,\overline{\alpha}_{\varepsilon}- \alpha)\le 0$ for all $\varepsilon>0,\ s\in[0,1/2)$ and $t\in(1/2,1]$ if and only if the condition in \autoref{prop:no_compromise} holds.
\end{proof}

\begin{proof}[Proof of the necessity part of \autoref{prop:interval}]
Suppose \autoref{LQ} holds and $c^*\in (0,1)$. Let $\alpha\in\mathcal{A}$ be the action rule induced by the delegation set $[c^*,1]$. We prove necessity of each condition in \autoref{prop:interval} in order.

\underline{Condition \ref{interval1}}: This follows from \autoref{lemma:interval_nes}.

\underline{Condition \ref{interval2}}: Fix $t\in (c^*/2,c^*)$ and $\varepsilon >0$. Let $a(\varepsilon)$ be the positive solution to $(c^*-t)a+a^2/2 = \varepsilon (t-c^*/2)$ and define $\overline{\alpha}_{\varepsilon}$ by
\begin{align*}
 \overline{\alpha}_{\varepsilon}(v)= \begin{cases}
 \alpha(v) &\text{ if } v\notin [c^*/2,c^*+a(\varepsilon)]\\
c^* - \varepsilon &\text{ if } v\in [c^*/2,t)\\
 c^* + a(\varepsilon) &\text{ if } v\in [t,c^*+a(\varepsilon)].
 \end{cases}  
 \end{align*}  
 Note that for any $\varepsilon>0$ small enough (so that $c^*-\varepsilon>c^*/2$ and $c^*+a(\varepsilon)<1$), $\overline{\alpha}_{\varepsilon}\in \mathcal{A}$. By definition of $a(\varepsilon)$, for any $v\in[0,1]$, $v\overline{\alpha}_{\varepsilon}(v)-\frac{[\overline{\alpha}_{\varepsilon}(v)]^2}{2}-\int_0^v\overline{\alpha}_{\varepsilon}(s)\mathrm ds \ge0$ and, hence, $\overline{\alpha}_{\varepsilon}$ is feasible for \eqref{e:relaxeda}. Therefore, if $\alpha$ is optimal, $\partial \phi(\alpha,\overline{\alpha}_{\varepsilon}-\alpha)\le0$. 
 
Note that
 \begin{align*}
 \partial \phi(\alpha,\overline{\alpha}_{\varepsilon}-\alpha) =  &- \varepsilon \int_{c^*/2}^{t} \left(u'(\alpha(v)) f(v) - \kappa f(v)\left[v-c^*-\frac{1-F(v)}{f(v)}\right]\right)\mathrm  dv \\
 &+ a(\varepsilon) \int_{t}^{c^*} \left(u'(\alpha(v))f(v) - \kappa f(v)\left[v-c^*-\frac{1-F(v)}{f(v)}\right]\right)\mathrm dv \\
 &+ \int_{c^*}^{c^*+a(\varepsilon)} (c^*+a(\varepsilon)-v) \left(u'(\alpha(v)) f(v) - \kappa f(v) \left[-\frac{1-F(v)}{f(v)}\right]\right) \mathrm dv.
 \end{align*}
By the implicit function theorem, $\lim_{\varepsilon\rightarrow 0}\frac{a(\varepsilon)}{\varepsilon}=\frac{t-c^*/2}{c^*-t}$. It follows that the last integral is of order $o(\varepsilon)$. Also, integration by parts implies that for $x,y\in\R$, $\int_x^y (f(v)(v-c^*)-[1-F(v)])\mathrm dv=F(y)(y-c^*)-F(x)(x-c^*)-y+x$. We conclude 
\begin{align*}
\lim_{\varepsilon \to 0_+} \frac{\partial\phi(\alpha,\overline{\alpha}_{\varepsilon}-\alpha)}{\varepsilon} 
 =& -u'(c^*)[F(t)-F(c^*/2)] + \kappa [F(t)(t-c^*)+F(c^*/2)(c^*-c^*/2)-t+c^*/2] \\
 &+ \frac{t-c^*/2}{c^*-t}\left[ u'(c^*)(F(c^*)-F(t)) - \kappa (c^*-t) (F(t)-1)  \right]\\
 =& - \frac{c^*-c^*/2}{c^*-t}u'(c^*)[F(t)-F(c^*/2)] + \kappa (c^*-c^*/2) [F(c^*/2)-F(t)]\\ 
 &+ \frac{t-c^*/2}{c^*-t} u'(c^*)(F(c^*)-F(c^*/2)) \\
  =&\frac{(c^*-c^*/2)(t-c^*/2)}{c^*-t}\\
  &  \times \left\{ - [u'(c^*)+\kappa (c^*-t)]\frac{F(t)-F(c^*/2)}{t-c^*/2} +  u'(c^*)\frac{F(c^*)-F(c^*/2)}{c^*-c^*/2}\right\}.
\end{align*}
Since $\partial \phi(\alpha,\overline{\alpha}_{\varepsilon}-\alpha)\le0$ for all $\varepsilon>0$, the last expression is negative for all $t\in(c^*/2,c^*)$, which implies condition \ref{interval2}.

\underline{Condition \ref{interval3}}: Fix $s\in[0,c^*/2)$ and $\varepsilon >0$. Let
\begin{align*}
 \overline{\alpha}_{\varepsilon}(v)= \begin{cases}
 0 &\text{ for } v<s\\
 \varepsilon &\text{ if } v\in [s,c^*/2)\\
 c^* - a(\varepsilon) &\text{ if } v\in [c^*/2,c^*-a(\varepsilon))\\
 v &\text{ if } v\ge c^*-a(\varepsilon),
  \end{cases}  
\end{align*}
 where $a(\varepsilon)\ge 0$ satisfies 
 \begin{align*}
 & c^*/2(c^*-a(\varepsilon))-(c^*-a(\varepsilon))^2/2-(c^*/2-s)\varepsilon=0\\
\iff & (c^*-c^*/2) a(\varepsilon) - (a(\varepsilon))^2=(c^*/2-s)\varepsilon.
 \end{align*}
For $\varepsilon$ small enough, there is a real solution $a(\varepsilon)\ge 0$, and a simple calculation shows that $\overline{\alpha}_\varepsilon$ is feasible for \eqref{e:relaxeda}.
Also, note that $\lim_{\varepsilon\rightarrow 0} \frac{a(\varepsilon)}{\varepsilon} = \frac{c^*/2-s}{c^*-c^*/2}$.

It follows from \eqref{eq:gateaux_nes1} that
 \begin{align*}
 \partial \phi(\alpha, \overline{\alpha}_{\varepsilon}- \alpha) = &\varepsilon\left[ \int_{s}^{c^*/2} u'(\alpha(v)) - \kappa\left[ v-\frac{1-F(v)}{f(v)} \right]  \mathrm dF(v) \right]\\
&- a(\varepsilon) \left[\int_{c^*/2}^{c^*} u'(\alpha(v)) - \kappa\left[ v-c^* -\frac{1-F(v)}{f(v)} \right] \mathrm dF(v) \right] + o(\varepsilon).
 \end{align*}
 Using integration by parts, we conclude
 \begin{align*}
 \lim_{\varepsilon \to 0_+} \frac{\partial\phi(\alpha,\overline{\alpha}_{\varepsilon}-\alpha)}{\varepsilon} &= u'(0) [F(c^*/2)-F(s)] - \kappa[c^*/2F(c^*/2)-sF(s)-c^*/2+s]\\
 &\ \ - \frac{c^*/2-s}{c^*-c^*/2}\left[ u'(c^*)[F(c^*)-F(c^*/2)] +\kappa(c^*-c^*/2)[1-F(c^*/2)]\right] \\
 &= (c^*/2-s) \left\{[u'(0)-\kappa s] \frac{F(c^*/2)-F(s)}{c^*/2-s} -  u'(c^*)\frac{F(c^*)-F(c^*/2)}{c^*-c^*/2} \right\}\le 0,
 \end{align*}
 which yields condition \ref{interval3}.
\end{proof}

\section{Proofs of Corollaries \ref{c:decr}, \ref{c:risk}, and \ref{c:LQinterval}}

\subsection{Proof of \autoref{c:decr}}
Since $u$ is concave, $u'$ is decreasing on $[0,1]$. Recall $\kappa\geq 0$. Hence, if the type density $f$ is decreasing on $[0,1]$, then $\kappa F-u'f$ is increasing on $[0,1]$. The result follows from \autoref{prop:full}.

\subsection{Proof of \autoref{c:risk}}
As $\kappa F(v)- u'(v) f(v)$ is continuous on $[0,1]$, it is increasing on $[0, 1]$ if its derivative is positive for all $v \in [0,1)$.  The derivative is $(\kappa-u^{\prime \prime}(v)) f(v)-u'(v)f'(v)$, which is larger than $u^{\prime \prime}(v) f(v)-u'(v)f'(v) $. The latter function is positive for all $v \in [0,1)$ if 
\[
\inf_{v \in [0, 1)} \frac{-u^{\prime \prime} (v)}{u'(v)} \geq \sup_{v \in [0, 1)} \frac{f'(v)}{f(v)}.
\]  
The RHS above is finite since $f$ is continuously differentiable and strictly positive on $[0, 1]$. Therefore, $\kappa F(v)- u'(v) f(v)$ is increasing on $[0,1]$ when the LHS above is sufficiently large. The result follows from \autoref{prop:full}.

\subsection{Proof of \autoref{c:LQinterval}}

Assume \autoref{LQ}. We prove the result by establishing that (i) logconcavity of $f$ on $[0,1]$ ensures that the conditions of either \autoref{prop:no_compromise} or \autoref{prop:interval} are satisfied, and (ii) if $\gamma>0$ (equivalently, given \autoref{LQ}, $u$ is strictly concave) or $f$ is strictly logconcave on $[0,1]$, then among interval delegation sets there is a unique optimum.

As introduced in \autoref{s:statics}, Proposer's expected utility from delegating the interval $[c, 1]$ with $c\in[0,1]$ is:
\begin{equation}
\label{e:delegate}
W(c)\equiv u(0)F(c/2)+u(c)(F(c)-F(c/2))+\int_{c}^{1}u(v) f(v)\mathrm dv.
\end{equation}
As shorthand for the function in condition \ref{interval1} of \autoref{prop:interval}, define 
\begin{equation}
\label{e:G}
G(v):=\kappa F(v)- u'(v) f(v).
%&=2\gamma F(v)-(1-\gamma+2\gamma(1-v))f(v)
\end{equation}

%We first establish some properties concerning quasiconvexity of $\kappa F- u' f$ (which appears in condition \ref{interval1} of \autoref{prop:interval}) and the $W$ function above.

We establish some properties of the $W$ and $G$ functions.

\begin{lemma} \label{l:qconvex}
Assume \autoref{LQ} and $f$ is logconcave on $[0,1]$. The functions $W$ and $G$ defined by \eqref{e:delegate} and \eqref{e:G} are respectively quasiconcave and quasiconvex on $[0,1]$, both strictly so if either $\gamma >0$ or $f$ is strictly logconcave on $[0,1]$. 
Furthermore, for any $c^* \in \arg\max_{c\in [0,1]} W(c)$, $G'(c^*/2) \leq 0$ if $c^*>0$ and $G'(c^*) \geq 0$ if $c^*<1$.
\end{lemma}  
\begin{proof}
The proof proceeds in four steps. Throughout, we restrict attention to the domain $[0,1]$ for the type density. Step 1 shows that $G$ is (strictly) quasiconvex and that $\{v:G'(v)=0\}$ is connected. Step 2 shows that $W$ can be expressed in terms of $G'$. Step 3 establishes that given any maximizer $c^*$ of $W$, $G$ is decreasing on $[0, c^*/2]$ and increasing on $[c^*, 1]$. Step 4 establishes the (strict) quasiconcavity of $W$. Note that under \autoref{LQ}, 
$\kappa\equiv \inf_{v\in[0,1)} -u^{\prime \prime}(v)=2\gamma$, $u'(v)=1-\gamma+2\gamma(1-v)$, and hence $G(v)=2\gamma F(v)-(1-\gamma+2\gamma(1-v))f(v)$.

\underline{Step 1}: We first establish that $G$ is (strictly) quasiconvex and that $\{v:G'(v)=0\}$ is connected. Logconcavity of $f$ implies that its modes (i.e., maximizers) are connected, and moreover $f'(v)=0\implies$ $v$ is a mode. Denote by $\mode$ the smallest mode. Since
\begin{equation}
\label{e:G'}
G'(v)=4\gamma f(v)-(1-\gamma+2\gamma(1-v))f'(v),
\end{equation}
it holds that
$\sign G'(v)=\sign \beta(v),$
where
\begin{equation*}
\beta(v):=4\gamma-\frac{f'(v)}{f(v)}(1-\gamma+2\gamma(1-v)).
\end{equation*}
On the domain $[0,\mode)$, $f'/f$ is positive and decreasing by logconcavity. Furthermore, $1-\gamma+2\gamma(1-v)$ is positive and decreasing. As the product of positive decreasing functions is decreasing, $\beta$ is increasing on the domain $[0,\mode)$. Since $\beta(v)\geq 0$ when $v\geq \mode$, it follows that $\beta$ is upcrossing (once strictly positive, it stays positive), and hence $G$ is quasiconvex.

We claim $\{v:\beta(v)=0\}$ is connected, which implies the same about $\{v:G'(v)=0\}$.  If $\gamma=0$ then $\beta(v)=0 \iff f'(v)=0$, which is a connected set, as noted earlier. If $\gamma>0$, then the conclusion follows because $\beta$ is increasing on $[0,\mode)$, $\beta(v)>0$ for $v>\mode$ (as $f'(v)\leq 0$), and $\beta$ is continuous. Furthermore, analogous observations imply that if either $f$ is strictly logconcave or $\gamma>0$, then $|\{v:G'(v)=0\}|\leq 1$ and so $G$ is strictly quasiconvex.

\underline{Step 2}:  We now show that 
\begin{equation}
\label{e:h_int_general}
W'(c)=\int_{c/2}^{c}(v-c)G'(v)\mathrm dv.
\end{equation}

The derivation is as follows:
\begin{align*}
W'(c)=&(F(c)-F(c/2))(1+\gamma-2\gamma c)-\frac{c}{2}f(c/2)(1+\gamma-\gamma c)\\
%=&(1+\gamma-2\gamma c)\left[(F(c)-F(c/2))-\frac{c}{2}f(c/2)\right]-\gamma \frac{c^2}{2}f(c/2)\\
=&(1+\gamma-2\gamma c)\left[\int_{c/2}^{c}f(v) \mathrm dv-\frac{c}{2}f(c/2)\right]-\gamma \frac{c^2}{2}f(c/2)\\
=&-(1+\gamma-2\gamma c)\int_{c/2}^{c}(v-c)f'(v) \mathrm dv-\gamma \frac{c^2}{2}f(c/2) \\
%=&\int_{c/2}^{c}-(v-c)(1+\gamma-2\gamma v)f'(v) \mathrm dv+2\gamma \int_{c/2}^{c}(v-c)^2f'(v) \mathrm dv -\gamma \frac{c^2}{2}f(c/2) \\
=&-\int_{c/2}^{c}(v-c)(1+\gamma-2\gamma v)f'(v) \mathrm dv+2\gamma \left[-\int_{c/2}^{c}(v-c)^2f'(v) \mathrm dv - \left(\frac{c}{2}\right)^2f(c/2)\right] \\
=&-\int_{c/2}^{c}(v-c)(1+\gamma-2\gamma v)f'(v) \mathrm dv+2\gamma \int_{c/2}^{c}2(v-c)f(v) \mathrm dv \\
=&\int_{c/2}^{c}(v-c)G'(v)\mathrm dv.
\end{align*}
The first equality above is obtained by differentiating \autoref{e:delegate} and using $u'(c)=1+\gamma-2\gamma c$ and $u(c)-u(0)=c(1+\gamma-\gamma c)$; the third and fifth equalities use integration by parts; the last equality involves substitution from \eqref{e:G'}; and the remaining equalities follow from algebraic manipulations.

\underline{Step 3}: We now establish that for any $c^* \in \arg\max\limits_{c\in [0,1]} W(c)$, $c^*>0 \implies G'(c^*/2) \leq 0$ and $c^*<1\implies G'(c^*) \geq 0$.

By Step 1, there exist $v_*$ and $v^*$ with $0 \leq v_* \leq v^* \leq 1$ such that $G'(v)<0$ on $[0, v_*)$, $G'(v)=0$ on $[v_*, v^*]$, and $G'(v)>0$ on $(v^*, 1]$. By \eqref{e:h_int_general}, $c \in (0, v_*) \implies$ $W'(c)>0$, and $c/2 \in (v^*, 1) \implies W'(c)<0$. Since $c^*$ is optimal, $c^*>0 \implies W'(c^*) \geq 0 \implies c^*/2 \leq v^* \implies G'(c^*/2) \leq 0$.  Similarly, $c^*<1 \implies W'(c^*)\leq 0 \implies c^* \geq v_* \implies G'(c^*) \geq 0$.

\underline{Step 4}: Finally we establish that $W$ is quasiconcave, strictly if $\gamma>0$ or $f$ is strictly logconcave. For this it is sufficient to establish that if $c>0$ and $W'(c)=0$, then $W^{\prime \prime}(c) \leq 0$, with a strict inequality if $\gamma>0$ or $f$ is strictly logconcave.

%
%Differentiating \eqref{e:FOCb}, $W^{\prime \prime}(c)$ is equal to 
%\begin{align*}
%&-2\gamma(F(c)-F(c/2))+(f(c)-f(c/2))(1+\gamma-2\gamma c)-\frac{c}{4}f'(c/2)(1+\gamma-\gamma c),
%\end{align*}
%which we can rewrite as 
%\[
%[2\gamma F(c/2)-f(c/2)(1+\gamma-\gamma c)]-[2\gamma F(c)-f(c)(1+\gamma-2\gamma c)]+\frac{c}{4}[4\gamma f(c/2)-f'(c/2)((1+\gamma-\gamma c))], 
%\]
%or equivalently
Differentiating \eqref{e:h_int_general},
\begin{align}
\label{e:H_prime}
W^{\prime \prime}(c) %&=2[2\gamma F(b)-f(b)(1+\gamma-2\gamma b)]-2[2\gamma F(2b)-f(2b)(1+\gamma-4\gamma b)]+b[4\gamma f(b)-f'(b)((1+\gamma-2\gamma b))] \\ \nonumber
&= \frac{c}{4}G'(c/2)-(G(c)-G(c/2)). 
\end{align}
Integrating by parts, 
\[
\int_{c/2}^{c}[(v-c)G'(v)+G(v)]\mathrm dv=\left[(v-c)G(v)\right]_{c/2}^{c}=\frac{c}{2}G(c/2).
\]
Now fix any $c>0$ such that $W'(c)=0$. By \eqref{e:h_int_general} and the above integration by parts, $
G(c/2)=(2/c)\int_{c/2}^{c}G(v)\mathrm dv$, 
which, because $G$ is quasiconvex by Step 1, implies 
$
G(c/2) \leq G(c), 
$
with a strict inequality if $\gamma>0$ or $f$ strictly logconcave. 
Similarly $G'(c/2) \leq 0$, and hence from \eqref{e:H_prime} we conclude that
$
W^{\prime \prime}(c) \leq 0, 
$
with a strict inequality  if $\gamma>0$ or $f$ is strictly logconcave. %Thus $W(c)$ is quasiconcave, strictly if $\gamma>0$ or $f$ is strictly logconcave. 
\end{proof}

We build on \autoref{l:qconvex} to establish \autoref{c:LQinterval} by verifying the conditions of \autoref{prop:no_compromise} and \autoref{prop:interval}.

\begin{proof}[Proof of \autoref{c:LQinterval}]
If the interval delegation set $[c^*, 1]$ is optimal then $c^*$ must maximize $W(c)$ defined in \eqref{e:delegate}.  Hence if $W$ is strictly quasiconcave---as is the case if $\gamma>0$ or $f$ is strictly logconcave on $[0,1]$, by \autoref{l:qconvex}---there can be at most one interval that is optimal.  So it suffices to establish that if $c^* \in \arg\max_{c\in [0,1]} W(c)$ then $[c^*, 1]$ is optimal.

To that end, we verify that if $c^*=1$ the conditions of  \autoref{prop:no_compromise} are satisfied and, if $c^*<1$, then conditions \ref{interval1}--\ref{interval3} of \autoref{prop:interval} are satisfied.  Note that condition \ref{interval1} is immediate %because $c^* > v^*$ by 
from \autoref{l:qconvex}.  As conditions \ref{interval2} and \ref{interval3} are vacuous for $c^*=0$ we need only consider $c^* \in (0, 1]$.  For any $c^* \in (0, 1)$ conditions \ref{interval2} and \ref{interval3} are jointly equivalent to
\[
\left(u'(0)-\kappa s \right) \frac{F(c^*/2)-F(s)}{c^*/2-s} \leq u'(c^*)\frac{F(c^*)-F(c^*/2)}{c^*/2} \leq \left(u'(c^*)+\kappa(c^*-t) \right) \frac{F(t)-F(c^*/2)}{t-c^*/2}
\]
for all $s\in [0,c^*/2)$ and $t\in (c^*/2,c^*]$. Substituting into the middle expression from the first-order condition $W'(c^*)=0$ (i.e., setting expression \eqref{e:opt_del} equal to zero and rearranging) yields
\begin{equation}
\label{e:inequality}
\hspace{-10pt}\left(u'(0)-\kappa s \right) \frac{F(c^*/2)-F(s)}{c^*/2-s} \leq (u(c^*)-u(0))\frac{f(c^*/2)}{c^*} \leq \left(u'(c^*)+\kappa(c^*-t) \right) \frac{F(t)-F(c^*/2)}{t-c^*/2}  
\end{equation}
for all $s \in [0,c^*/2)$ and $t\in (c^*/2,c^*]$.  So if \eqref{e:inequality} holds for $c^* \in (0, 1)$ then the conditions in \autoref{prop:interval} are verified. On the other hand, since the condition in \autoref{prop:no_compromise} is equivalent to the right-most term in \eqref{e:inequality} being larger than the left-most term for all $s \in [0,c^*/2)$ and $t\in (c^*/2,c^*]$ when $c^*=1$, \eqref{e:inequality} holding for $c^*=1$ implies the condition in \autoref{prop:no_compromise}. Accordingly, we fix a $c^*>0$ and verify the two inequalities of \eqref{e:inequality} in turn.

\underline{First inequality of \eqref{e:inequality}}: %We first show that 
%\[
%(u'(0)-\kappa s) \frac{F(c^*/2)-F(s)}{c^*/2- s} \leq (u(c^*)-u(0))\frac{f(c^*/2)}{c^*} 
%\] 
%for all $s \in [0,c^*/2)$.
Using $u'(a)=1+\gamma-2\gamma a$, 
$\kappa=2\gamma$, and $\frac{u(a)-u(0)}{a}=1+\gamma-\gamma a$, the first inequality of \eqref{e:inequality} reduces to
\[
(1+\gamma-2\gamma s) \frac{F(c^*/2)-F(s)}{c^*/2- s}  \le (1+\gamma-\gamma c^*)f(c^*/2) \quad \forall s\in [0,c^*/2).
\]   

It follows from L'Hopital's rule that the above inequality holds with equality in the limit as $s \rightarrow c^*/2$.  Hence it is sufficient to demonstrate that the LHS of the inequality is increasing 
%\[
%\frac{\partial}{\partial s} \left[(1+\gamma-2\gamma s)\frac{F(c^*/2)-F(s)}{c^*/2- s}\right] \geq 0
%\]
for all $s \in [0, c^*/2)$.  For any $s \in [0, 1]$ let
\begin{equation}
\label{e:D}
D(s):=(1+\gamma-\gamma c^*)(F(c^*/2)-F(s))-(c^*/2-s)(1+\gamma -2\gamma s)f(s),  
\end{equation}
and observe that 
\[
\frac{\partial}{\partial s} \left[(1+\gamma-2\gamma s)\frac{F(c^*/2)-F(s)}{c^*/2- s}\right]=\frac{1}{(c^*/2-s)^2}D(s).
\]
So it is sufficient to show that, for all $s \in [0, c^*/2)$, $D(s)\geq 0$. This holds because $D(c^*/2)=0$ and, for all $s <c^*/2$,
\begin{align}
D'(s)%&=-(1+\gamma-\gamma c^*)f(s)-(c^*/2-s)(1+\gamma -2\gamma s)f'(s)-(4 \gamma s-1-\gamma-\gamma c^*)f(s) \\
&=(c^*/2-s)[4\gamma f(s)-(1+\gamma -2\gamma s)f'(s)] \quad \text{ \small{differentiating \eqref{e:D} and simplifying} } \notag \\
&=(c^*/2-s)G'(s) \quad \text{ \small{substituting from \eqref{e:G'}} } \label{e:D'}\\
& \leq 0 \quad \text{ \small{by \autoref{l:qconvex}}}.\notag
\end{align}

\underline{Second inequality of \eqref{e:inequality}}: Using $u'(a)=1+\gamma-2\gamma a$, $\kappa=2\gamma$, and $\frac{u(a)-u(0)}{a}=1+\gamma-\gamma a$, the second inequality of \eqref{e:inequality} reduces to
\[
(1+\gamma-\gamma c^*)f(c^*/2) \le (1+\gamma-2\gamma t) \frac{F(t)-F(c^*/2)}{t-c^*/2} \quad \forall t \in (c^*/2,c^*].
\] 
Using L'Hopital's rule for the limit as $t \rightarrow c^*/2$ and the fact that $W'(c^*) \geq 0$ by optimality of $c^*>0$, it follows that 
\[
\hspace{-10pt}\lim_{t \rightarrow c^*/2} (1+\gamma - 2\gamma t)\frac{F(t)-F(c^*/2)}{t-c^*/2}=(1+\gamma - \gamma c^*)f(c^*/2) \leq (1+\gamma - 2\gamma c^*)\frac{F(c^*)-F(c^*/2)}{c^*/2}.
\]
Hence it is sufficient to show that $(1+\gamma - 2\gamma t)\frac{F(t)-F(c^*/2)}{t-c^*/2}$ is quasiconcave for $t \in (c^*/2, c^*]$.  Note that % the derivative of $(1+\gamma - 2\gamma t)\frac{F(t)-F(c^*/2)}{t-c^*/2}$  is
\[
\frac{\partial}{\partial t} \left[(1+\gamma-2 \gamma t)\frac{F(t)-F(c^*/2)}{t-c^*/2}\right]=
\frac{1}{(t-c^*/2)^2}D(t),%[\underbrace{(t-c^*/2)(1+\gamma -2\gamma t)f(t)-(1+\gamma-\gamma c^*)(F(t)-F(c^*/2))}_{\textstyle{=:R(t)}}], 
\]
where $D$ is defined in \eqref{e:D}, and so 
\[
\sign\frac{\partial}{\partial t} \left[(1+\gamma -2\gamma t)\frac{F(t)-F(c/2)}{t-c/2}\right]=\sign D(t).
\]
Since $D(c^*/2)=0$, it follows that $(1+\gamma-2\gamma t)\frac{F(t)-F(c^*/2)}{t-c^*/2}$ is quasiconcave for $t \in (c^*/2, c^*]$ if $D$ is quasiconcave.  $D$ is quasiconcave because, as was shown in \eqref{e:D'}, $D'(t)=(c^*/2-t)G'(t)$,  which is positive then negative on $(c^*/2, c^*]$ by the quasiconvexity of $G$ (\autoref{l:qconvex}).
\end{proof}

\section{Proof of \autoref{prop:compstats}}
\begin{proof}[Proof of \autoref{prop:compstats}\ref{p:risk}] 
Let $H(a,c)$ denote the cumulative distribution function of the action implemented under the interval delegation set $[c,1]$. That is,
\begin{align*}
H(a,c)= \begin{cases}
0  & \text{ if } a<0 \\
F(c/2) & \text{ if } 0\le a < c \\
F(a) & \text{ if } c\le a < 1 \\
1  & \text{ if } 1\le a.
\end{cases}
\end{align*}

For any $0\leq c_L <c_H\leq 1$ the difference $H(\cdot,c_L)-H(\cdot,c_H)$ is upcrossing: once strictly positive, it stays positive. It follows from the variation diminishing property \citep[Theorem 3.1 on p.~21]{Karlin:68} that 
$$\int_0^1 \hat u'(a)\left[H(a,c_L)-H(a,c_H)\right]\mathrm da \geq (>) 0 \implies \int_0^1 u'(a)\left[H(a,c_L)-H(a,c_H)\right]\mathrm da \geq (>) 0$$
when $\hat u$ is strictly more risk averse than $u$. Integrating by parts, we obtain 
$$\int_0^1 \hat u(a)\left[H(\mathrm da,c_L)-H(\mathrm da,c_H)\right] \leq (<) 0 \implies \int_0^1 u(a)\left[H(\mathrm da,c_L)-H(\mathrm da,c_H)\right] \leq (<) 0.$$
 A standard monotone comparative statics argument \citep{MS:94} then implies that $C^*(u)\geq_{SSO} C^*(\hat u)$.
\end{proof}

\begin{proof}[Proof of \autoref{prop:compstats}\ref{p:LR}]
Let density $f(v)$ strictly dominate density $g(v)$ in likelihood ratio on the unit interval: i.e., for all $0\leq v_L<v_H\leq 1$, $f(v_L)g(v_H)<f(v_H)g(v_L)$.  Let $w(c,v)$ denote Proposer's payoff under the interval delegation set $[c,1]$ when Vetoer's type is $v$. We have
$$
w(c,v)=
\begin{cases} 
u(0) & \text{ if } v<c/2 \\
u(c) & \text{ if } v\in (c/2,c) \\
u(v) &\text{ if } v \in (c,1). \\
u(1) &\text{ if } v > 1. \\
\end{cases}
$$
Consider any $0\leq c_L <c_H\leq 1$. The difference $w(c_H,\cdot)-w(c_L,\cdot)$ is upcrossing: once strictly positive, it stays positive. It follows from the variation diminishing property \citep[Theorem 3.1 on p.~21]{Karlin:68} that $$ 
\int_0^1 \left[w(c_H,v)-w(c_L,v)\right]g(v)\mathrm d v \geq (>)0 \implies \int_0^1 \left[w(c_H,v)-w(c_L,v)\right]f(v)\mathrm d v \geq (>) 0.$$ A standard monotone comparative statics argument \citep{MS:94} then implies that $C^*(f)\geq_{SSO} C^*(g)$.
\end{proof}

\section{Proof of \autoref{p:pareto}}
Let $a_U$ and $a_I$ denote proposals in some noninfluential and influential cheap-talk equilibria, respectively (the latter may not exist). It is straightforward that $a_U>0$ and, if it exists, $a_I\in (0,1)$. Since \autoref{p:pareto}'s conclusion is trivial for full delegation ($c^*=0$), it suffices to establish that any optimal interval delegation set $[c^*,1]$ with $c^*\in (0,1)$ has $c^*<\min\{a_I,a_U\}$. (By convention, $\min\{a_I,a_U\}=a_U$ if $a_I$ does not exist.)

Plainly, $a_U$ is a noninfluential equilibrium proposal if and only if
 \[
a_U \in \argmax_{a} \left[ u(0)F(a/2)+u(a)(1-F(a/2))\right],
 \]
and so if $a_U<1$ then it solves the first-order condition
 \begin{equation}
\label{e:uninform}
{2u'(a)}\left[{1-F(a/2)}\right]-{f(a/2)}\left[{u(a)-u(0)}\right]=0.
\end{equation} 

Any influential cheap-talk equilibrium outcome can be characterized by a threshold type $v_I\in (0,1)$ such that types $v<v_I$ pool on the ``veto threat'' message, and types $v>v_I$ pool on the ``acquiesce'' message. Since type $v_I$ must be indifferent between sending the two messages, and she will accept either proposal from the Proposer, it holds that 
$$v_I=\frac{1+a_I}{2}.$$
It follows that $a_I$ is an influential equilibrium proposal if and only if
\begin{equation*}
a_I \in \argmax_{a} \left[u(0)\frac{F(a/2)}{F((1+a_I)/2)}+u(a)\left(1-\frac{F(a/2)}{F((1+a_I)/2)}\right)\right].
\end{equation*}
The first-order condition is that function \eqref{e:influential} in the main text equals zero, i.e., $a_I\in (0,1)$ solves
\begin{equation}
\label{e:influential_2}
{2u'(a)}\left[{F((1+a)/2)-F(a/2)}\right]-{f(a/2)}\left[{u(a)-u(0)}\right]=0.
\end{equation}
Note that at $a=0$, the LHS is strictly positive. Hence, if the LHS is strictly downcrossing on $(0,1)$, then \autoref{e:influential_2} has at most one solution in that domain; if there is a solution, then \autoref{e:influential_2}'s LHS is strictly positive (resp., strictly negative) to its left (resp., right); furthermore, it can be verified that the solution then identifies an influential equilibrium. Note that if there is no solution to \autoref{e:influential_2} on $(0,1)$ then there is no influential equilibrium. 

Turning to optimal interval delegation, recall from
 \autoref{s:statics} that the threshold is a zero of the function \eqref{e:opt_del}, i.e., $c^*\in (0,1)$ solves
\begin{equation}
\label{e:opt_del_2}
%\frac{2u'(a)}{u(a)-u(0)}&=\frac{f(a/2)}{F(a)-F(a/2)}.
{2u'(a)}\left[{F(a)-F(a/2)}\right]-{f(a/2)}\left[{u(a)-u(0)}\right]=0.
\end{equation}
If the LHS is strictly downcrossing on $(0,1)$, then on that domain $c^*$ is the unique solution to \autoref{e:opt_del_2} and \autoref{e:opt_del_2}'s LHS is strictly positive (resp., strictly negative) to the solution's left (resp., right).

For any $a\in (0,1)$ the LHS of \autoref{e:uninform} is strictly larger than the LHS of \autoref{e:influential_2}, which in turn is strictly larger than the LHS of \autoref{e:opt_del_2}. 
If there is no solution in $(0,1)$ to \autoref{e:influential_2}, then its LHS is always strictly positive, and hence there are neither any influential equilibria nor any noninfluential equilibria with $a_U<1$, and we are done.  So assume at least one solution in $(0,1)$ to \autoref{e:influential_2}. Let 
\begin{align*}
\underline a_2&:=\inf\{a\in (0,1):\text{ \autoref{e:influential_2}'s  LHS } \leq 0\},\\
\overline a_2&:=\sup\{a\in (0,1):\text{ \autoref{e:influential_2}'s  LHS } \geq 0\},
\end{align*}
and analogously define $\underline a_3$ and $\overline a_3$ using \autoref{e:opt_del_2}'s LHS. The aforementioned ordering of the equations' LHS, \autoref{e:influential_2}'s LHS being strictly positive at 0, and continuity combine to imply $0\leq \underline a_3<\underline a_2<a_U$, and $\overline a_3\leq \overline a_2$ with a strict inequality if either $\overline a_2<1$ or $\overline a_3<1$. Furthermore, $a_I\in[\underline a_2,\overline a_2]$ and $c^*\in [\underline a_3,\overline a_3]$.

If the LHS of \autoref{e:influential_2} is strictly downcrossing on $(0,1)$, then by the properties noted right after \autoref{e:influential_2}, $a_I=\underline a_2=\overline a_2<1$ and hence $c^*<\min\{a_I,a_U\}$. If the LHS of \autoref{e:opt_del_2} is strictly downcrossing on $(0,1)$, then by the properties noted right after \autoref{e:opt_del_2}, $c^*=\underline a_3=\overline a_3<1$ and hence $c^*<\min\{a_I,a_U\}$.
 
 \section{Stochastic Mechanisms can be Optimal}
 \label{s:stochastic}
\begin{example}
\label{eg:stochastic}

Suppose Proposer has a linear loss function, $\underline v=0$, $\overline v=1$, and $f(v)$ is strictly increasing except on $(1/2- \delta,1/2+ \delta)$, where it is strictly decreasing. Assume $|f'(v)|$ is constant (on $[0,1]$).\footnote{As it is nondifferentiable at two points, this density violates our maintained assumption of continuous differentiability. But the example could straightforwardly be modified to satisfy that assumption.} Take $\delta>0$ to be small. See \autoref{fig:stochastic}.

\begin{figure}
\centering
  \begin{tikzpicture}[] 
         \draw[axis,->] (0,-1) --(5.2,-1)node[right]{$v$};
            \draw (5,-1) node[below]{1};
      \draw (0,-1) node[below]{0};
      \draw (2.5,-1) node[below]{$\frac{1}{2}$};
         \draw[axis,->] (0,-1) --(0,3.2) ;
         \draw[blue] (0,0) -- (2.3,1.5) -- (2.7,1.3) -- (5,2.8) node[right] {\textcolor{blue}{$f$}};
  \end{tikzpicture}
\caption{A density under which no compromise is the optimal delegation set when Proposer has a linear loss function, but it is worse than some stochastic mechanism.}
\label{fig:stochastic}
\end{figure}
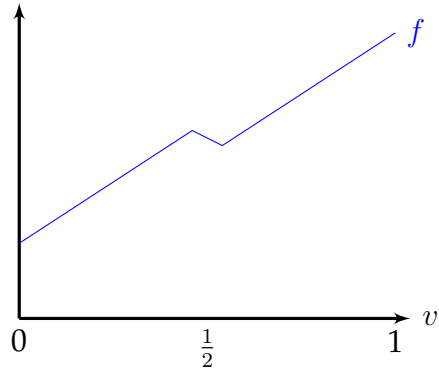

Recall that if $\delta$ were $0$, then no compromise (i.e., the singleton menu $\{1\}$) would be optimal by \autoref{prop:no_compromise} or the discussion preceding it. It can be verified that no compromise remains an optimal delegation set for small $\delta>0$. %\footnote{Our second-order stochastic dominance argument show that the optimal delegation set contains at most two actions in $[0,1/2- \delta]$, at most one interval and no isolated points in $[1/2- \delta, 1/2+ \delta]$, and at most two actions in $[1/2+ \delta,1]$. Suppose a delegation set contains actions $r<s<t$ and $r$ is the largest action below $s$ and $t$ is the smallest action above $s$. Then $P$'s expected payoff increases by removing $s$ from the delegation set if  \[\frac{F(\frac{t+s}{2})-F(\frac{t+r}{2})}{\frac{s-r}{2}}\ge \frac{F(\frac{t+r}{2})-F(\frac{s+r}{2})}{\frac{t-s}{2}}.\] This implies that if $\delta>0$ is small enough, any optimal delegation set cannot contain an isolated action different from 0 or 1. It then follows that the optimal delegation set is $\{0,1\}$.} 
We argue below that Proposer can obtain a strictly higher payoff, however, by adding a stochastic option $\ell$ that has expected value $1/2$ and is chosen only by types in $(1/2- \delta, 1/2 + \delta)$.

The stochastic option $\ell$ provides action $1-\frac{1}{2p}$ with probability $p$ and action $1$ with probability $1-p$. For any $p\in (0,1)$, this lottery has expected value $1/2$. Moreover, when $p=\frac{1}{2- 4 \delta}$, quadratic loss implies that type $1/2- \delta$ is indifferent between $\ell$ and action $0$ while 
% \begin{align*}
%   u(0,1/2- \delta) &= - (1/2 - \delta)^2\\
%   u(\ell,1/2- \delta) &= - p\left(1/2 - \delta - 1 + \frac{1}{2p}\right)^2 - (1-p)(1/2 - \delta-1)^2
% \end{align*}
type $1/2+ \delta$ is indifferent between $\ell$ and $1$. Consequently, any type in $[0,1/2-\delta)$ strictly prefers $0$ to both $\ell$ and $1$; any type in $(1/2- \delta, 1/2 + \delta)$ strictly prefers $\ell$ to both $0$ and $1$; and any type in $(1/2+\delta,1]$ strictly prefers $1$ to both $\ell$ and $0$.

Therefore, offering the menu $\{\ell,1\}$ rather than $\{1\}$ changes the induced expected action from $0$ to $1/2$ when $v\in (1/2- \delta, 1/2)$ and from $1$ to $1/2$ when $v\in (1/2, 1/2+\delta)$. Since $f(v)$ is strictly decreasing on $(1/2- \delta,1/2+ \delta)$, Proposer is strictly better off. Note that if one were to replace $\ell$ with a deterministic option that provides $\ell$'s expected action $1/2$, then all types in $(1/4,3/4)$ would strictly prefer to choose that option over both $0$ and $1$. So the menu $\{1/2,1\}$ is strictly worse than not only $\{\ell,1\}$ but also just $\{1\}$.\footnote{Vis-\`a-vis \autoref{lem:relaxed-stochastic} and its proof that involves replacing a stochastic mechanism with its ``averaged'' deterministic counterpart: in this example the deterministic mechanism that solves problem \eqref{e:relaxeda} cannot be incentive compatible. In particular, it is not a mechanism corresponding to any delegation set.} \marker
\end{example}
\end{document}